\documentclass[reqno, 11pt]{amsart}%
\usepackage{amssymb}
\usepackage[nesting]{hyperref}
\usepackage[pdftex]{graphicx}
\usepackage{listings}
\usepackage{multirow}
\usepackage{placeins}
\usepackage{color}
\usepackage{subfigure}
\usepackage{lscape}
\usepackage{dsfont}
\usepackage{amsmath}
\usepackage{amsfonts}%
\usepackage{tikz}
\usepackage{verbatim}
\providecommand{\U}[1]{\protect\rule{.1in}{.1in}}
\newcommand{\BlackBoxes}{\global\overfullrule5pt}

\BlackBoxes

\newcommand{\R}{\mathbb{R}}
\newcommand{\N}{\mathbb{N}}

\newcommand{\Eop}{\mathbb{E}}
\newcommand{\Pop}{\mathbb{P}}
\newcommand{\Q}{\mathbb{Q}}

\newtheorem{theorem}{Theorem}
\newtheorem{corollary}[theorem]{Corollary}
\newtheorem{lemma}[theorem]{Lemma}

\theoremstyle{definition}
\newtheorem{example}[theorem]{Example}
\newtheorem{remark}[theorem]{Remark}
\newtheorem{definition}[theorem]{Definition}
\numberwithin{equation}{section} \numberwithin{theorem}{section}
\def\0{\kern0pt\-\nobreak\hskip0pt\relax}

\makeatletter
\AtBeginDocument{ \def\@serieslogo{ \vbox to\headheight{ \parindent\z@ \fontsize{6}{7\p@}\selectfont
\vss}}}

\def\makeoverbar#1#2#3#4#5#6#7{ \setbox0=\hbox{$\m@th#2\mkern#5mu{{}#3{}}\mkern#6mu$} \setbox1=\null \dimen@=#4\fontdimen8#13 \dimen@=3.5\dimen@
\advance\dimen@ by \ht0 \dimen@=-#7\dimen@ \advance\dimen@ by \wd0
\ht1=\ht0 \dp1=\dp0 \wd1=\dimen@
\dimen@=\fontdimen8#13 \fontdimen8#13=#4\fontdimen8#13
\rlap{\hbox to \wd0{$\m@th\hss#2{\overline{\box1}}\mkern#5mu$}}
\fontdimen8#13=\dimen@}
\def\mylabel#1#2{{\def\@currentlabel{#2}\label{#1}}}
\makeatother

\begin{document}
\title[Consistent upper price bounds for exotic options ]{Consistent upper price bounds for exotic options given a finite number of call prices and their convergence}
\author[N. \smash{B\"auerle}]{Nicole B\"auerle${}^*$}
\address[N. B\"auerle]{Department of Mathematics,
Karlsruhe Institute of Technology (KIT), D-76128 Karlsruhe, Germany}

\email{\href{mailto:nicole.baeuerle@kit.edu}
{nicole.baeuerle@kit.edu}}

\author[D. \smash{Schmithals}]{Daniel Schmithals${}^*$}
\address[D. Schmithals]{Department of Mathematics,
Karlsruhe Institute of Technology (KIT), D-76128 Karlsruhe, Germany}

\email{\href{daniel.schmithals@kit.edu} {daniel.schmithals@kit.edu}}

\thanks{${}^*$ Department of Mathematics,
Karlsruhe Institute of Technology (KIT), D-76128 Karlsruhe, Germany}
\begin{abstract}
We consider the problem of finding a consistent upper price bound for exotic options whose payoff depends on the stock price at two different predetermined time points (e.g. Asian option), given a finite number of observed call prices for these maturities. A model-free approach is used, only taking into account that the (discounted) stock price process is a martingale under the no-arbitrage condition. In case the payoff is directionally convex we obtain the worst case marginal pricing measures. The speed of convergence of the upper price bound is determined when the number of observed stock prices increases. We illustrate our findings with some numerical computations. 
\end{abstract}
\maketitle


\makeatletter \providecommand\@dotsep{5} \makeatother



\vspace{0.5cm}
\begin{minipage}{14cm}
{\small
\begin{description}
\item[\rm \textsc{ Key words} ]
{\small Martingale Optimal Transport; Directional Convexity; Convex Order; Asian Option}
\end{description}
}
\end{minipage}

\section{Introduction}
Given a finite number of observable call prices on the same stock for two different maturities $0<t_1<t_2$ and different strikes, what is an arbitrage-free upper price bound for an arbitrary option whose payoff is a function of the stock price at time $t_1$ and $t_2$? A typical example for such an option would be an Asian option. We study here a model-free setting, only relying on the assumption that the discounted stock price process is a martingale. In case the option's payoff function is directionally convex we show that the upper bound for the price is given by the optimal martingale transport between the marginal distributions which are obtained from  linearly interpolating the observable call prices. Moreover, it seems intuitively clear that when the number of observed call prices for different strikes increases that the so constructed upper bound converges against the true upper bound. We show this and also determine the best possible speed of convergence.

In case the marginal risk neutral distributions of the stock are completely known, the problem of finding an upper bound for the price of another derivative which is only a function of the two stock prices is know as martingale optimal transport problem (see e.g. \cite{blp13}). More precisely, the {\em martingale optimal transport problem} is the problem to maximize (minimize) 
\begin{equation}\label{eq:cQprob}
\int_{\R^2} c(x,y) \Q(dx,dy)
\end{equation}
under the constraints that the margins of $\Q$ are predefined distributions $\mu,\nu$, e.g. $\Q(dx,\R)=\mu(dx)$ and $\Q(\R,dy)=\nu(dy)$ and $\mu$ and $\nu$ are the distributions of a martingale, i.e.
\begin{equation}
\int_{\R} y\Q(x,dy)= x, \mbox{ for } \mu-a.e. \; x\in \R.
\end{equation}
According to the Lemma of Breeden and Litzenberger (see \cite{bl}) the risk neutral marginal distribution of a stock price at time $t_i$ can be obtained when call prices of all strikes for the maturity $t_i$ are observable. Assuming that the market is free of arbitrage this then leads to the martingale optimal transport problem (for more details, see \cite{blp13}). For special options, these kind of problems have already been discussed in \cite{h98} using Skorokhod embedding techniques. In case the payoff function $c$ satisfies certain properties (more precisely the Martingale Spence Mirrlees condition: $c_{xyy}\ge0$), it has been shown in \cite{hlt16} (see also \cite{bj16}) that the optimal martingale transport is a so-called left-monotone transport plan. In case the marginal distributions are discrete, the problem can be formulated as a linear programme.  In \cite{go17}, the authors approximate the general optimal martingale problem by a sequence of linear programmes, obtained by  discretization of the marginal distributions coupled with a suitable relaxation of the martingale constraint. The convergence rate of this approximation is also obtained.

In this paper however, we assume that there are only a finite number of call prices for both maturities $0<t_1<t_2$ observable which seems to be the realistic case. But then a whole family of risk neutral marginal distributions for both times points are available. In combination with the  martingale condition which is then the best upper price bound? I.e. we want to maximize \eqref{eq:cQprob} over all $\Q$ which satisfy the martingale condition and are consistent with observed call prices. Unlike in classical transport problems we cannot separate the margins from the copula since they are connected via the martingale condition.  However, we are able to answer the question when the  payoff function is directionally convex, i.e. convex in the components and supermodular on $\R^2$. In this case the worst margins are the ones which are obtained from the linearly interpolated call price function. 

In an older stream of literature, see e.g. \cite{r94,d94,dk98} researchers have calibrated discrete models to given option prices by constructing implied trees. The calibration is to a volatility surface obtained from interpolated market data. But the question of no arbitrage has not been discussed at that point in time directly. Moreover, in \cite{sst03} it has been shown that calibrating given models to data might lead to fairly different prices for exotic options.

On the other hand a lot of studies are concerned with price bounds for specific options. For example in \cite{hlw05} lower bounds for prices of basket options (with two stocks) are derived, given a finite number of observations of call prices of all stocks in the basket. The paper \cite{LW08} considers upper bounds for spread options which are basket options where the weights may have arbitrary sign. The authors use call price information of all strikes and special properties of the payoff function to derive the inequality.  Since only one time point is important in the payoff of basket options, the martingale property does not give further information in this case. 

There are also a number of papers which consider bounds and semi-static hedging strategies for Asian  options. Here the option payoff depends on the performance of one stock at different time points and the martingale property as a further information is highly relevant. Moreover, the payoff of an Asian option like
\begin{equation}
\Big( \frac12(S_{t_1}+S_{t_2})-k\Big)^+
\end{equation} 
as the average of the stock price at two time points, satisfies the directional convexity needed for our main result.  In \cite{cdcv08} options written on weighted sums of asset prices are considered. The study includes basket and Asian options. Upper bounds and super-replication strategies for these kind of options are derived in the case that all relevant call prices on the options are observed and in the case of a finite number of call prices are given. The convergence issue is also treated. The authors use comonotonicity arguments to construct the upper bound and generalize results in \cite{sim00}. This however is (in genral) in contradiction to the martingale property (see also Remark 2.3 in \cite{ams08}). I.e. respecting the martingale property should lead to tighter bounds. In \cite{ams08} lower bounds on Asian options are derived under some further assumptions on the expectation of the stock price process which are shown to be satisfied in L\'evy markets. The case of a finite number of observable call prices is also considered. A completely different approach is pursued in \cite{ck17} where an Asian option with  continuous stock price average is considered and bounds are derived using dynamic programming techniques. 

Our paper is organized as follows:  In the next section we summarize some facts about the relation between call prices and pricing measures and about the convex order. In Section \ref{sec:margin} we construct a special distribution from a finite number of observable call prices, show that is a maximal element with respect to the convex order in the set of all consistent pricing measures and determine its Wasserstein distance to other consistent pricing measures. Section \ref{sec:bounds} identifies the worst case margins, given the payoff function is directionally convex. The next section determines the best possible convergence rate of the upper bound, given the number of observable call prices increases. Section \ref{sec:numerics} provides some numerical results for the convergence and one example for the calculation of an upper price bound for an Asian option using real data. The appendix contains some longer proofs and an example showing that the speed of convergence is the optimal one in general.

\section{Preliminary Results }
Consider a financial market with one risk-free asset and one risky asset.
We consider only two future time points which we denote by $0<t_1<t_2$.  The risk-free asset has no interest and is normalized by $1$ and for the risky asset with price process $(S_t)$ we write $(S_{t_1},S_{t_2})=(X,Y)$ and assume $S_0 =1$. The random variables $X,Y$ are non-negative and defined on a probability space $(\Omega,\mathcal{F},\Pop)$. No model for the stock price is assumed, but throughout - and this is important - we assume that the market is free of arbitrage.

\subsection{Call Options and Pricing Measures}
Using the assumption of no-arbitrage  we may derive some properties of the price function $k \mapsto C_i(k)$, where $C_i(k)$ shall represent the price at time $t=0$ of a call option with strike price $k \in \R_+$ and maturity $t_i, i=1,2$ written on the stock.  More precisely $C_i$ should have the following properties:

  \begin{lemma}\label{Lem Call Prop 1}
    The mapping $C_i:\R_+ \to \R_+$, $k \mapsto C_i(k), i=1,2$ is  in an arbitrage free market
        \begin{enumerate}
        \item decreasing.
        \item convex.
      \item $\lim_{k\to \infty} C_i(k)=0$.
      \item $C'_i(0+) \geq -1$.
      \item $C_i(0)=S_0=1$.
    \end{enumerate}
  \end{lemma}

A proof can e.g. be found in \cite{dh07} or \cite{schm18} Lemma 3.2 and Lemma 3.4 or in \cite{dh07}, Theorem 9.1. 

\begin{definition}
    \label{Def Candidate}  
    A function $C:\R_{+} \to \R_{+}$ is called a \textit{candidate function for call option prices}, if it satisfies conditions (a)-(e) from Lemma \ref{Lem Call Prop 1}.
  \end{definition}
  
The no-arbitrage condition implies by virtue of the first fundamental theorem of asset pricing in a general market the existence of a pricing measure.  More precisely, let us denote by $P(\R_+)$ the set of all probability measures on $\R_+$. Then there exists a measure $\mu\in P(\R_+)$ such that 
 \begin{equation}\label{eq:Cfunc}
C(k) = \int (x-k)^+ \mu(dx), \quad k\ge 0.
\end{equation}   In this case we say that $\mu$ and $C$ are {\em consistent.}   
The lemma of Breeden and Litzenberger \cite{bl} states that in this situation
\begin{equation}\label{eq:BL} \mu((-\infty,x]) = 1+C'(x+), \quad x\in\R.\end{equation}
That means when we know call prices for all strikes $k>0$ (on the same stock, with same maturity), we can derive the pricing measure $\mu$ by \eqref{eq:BL}.
Note that contrary if $\mu\in P(\R_+)$ is given, the function $C$ defined in \eqref{eq:Cfunc} is automatically a candidate function if $\int x\mu(dx)=S_0=1.$

\subsection{Convex Order}  
Another important tool that we need is the so-called convex order.
 \begin{definition}\label{Def Convex Order}
    Two  measures $\mu,\nu $ on $\R$  are said to be in \textit{convex order}, denoted by $\mu\leq_{c}\nu$, if for any convex function $f:\R \to \R$ such that the integrals exist,
    \[
      \int_{\R} f(x) \mu( d x) \leq \int_{\R} f(x) \nu( d x).
    \]
      \end{definition}
Since both $f(x)=x$ and $f(x)=-x$ are convex as well as $f(x)=1$ and $f(x)=-1$, the property $\mu\leq_c\nu$ implies that $\int x\mu(dx)=\int x\nu(dx)$ and $\mu(\R)=\nu(\R)$. The next result follows from \cite{str}:

\begin{lemma}\label{lem:cx1}
Suppose $\mu,\nu \in P(\R_+).$ Then $\mu\le_c \nu$ is equivalent to the existence of a probability space $(\Omega,\mathcal{F},\Pop)$ and non-negative random variables $X,Y$ on it such that $X$ has distribution $\mu$ and $Y$ has distribution $\nu$ and $X = \Eop[Y|X].$ 
\end{lemma}

The next lemma follows from Theorem 1.5.3 and Theorem 1.5.7 in \cite{ms}. 

 \begin{lemma}\label{lem:cx2}
    Let $\mu, \nu\in P(\R_+) $  and denote by $C_\mu$ and $C_\nu$ the respective  consistent pricing functions. Suppose that $\int x\mu(dx)=\int x\nu(dx)=1$. Then $\mu \leq_c  \nu$ is  equivalent to $C_{\mu} \leq  C_{\nu}$.
\end{lemma}

\section{Specially Designed Marginals}\label{sec:margin}
The lemma  of Breeden and Litzenberger (see \cite{bl}) implies that when we can observe call prices for a fixed maturity  for all strikes $k>0$, we are able to determine the pricing measure which in turn defines prices of other European options with the same maturity.
In reality  however, there are for a fixed maturity $t_i$ only a finite number of call prices $c_0^i>\ldots >c_{n_i}^i>0$ available for strikes $0\le k_0^i<\ldots <k_{n_i}^i, n_i\in\N, i=1,2$ . Thus we define
  
  \begin{definition}    \label{Def Candidate2}  
    Let  for $i=1,2$:
    $$ P_i:= \Big\{ \mu\in P(\R_+) : c_j^i= \int (x-k_j^i)^+ \mu(dx), j=0,\ldots ,n_i, \int x\mu(dx)=1\Big\}$$
    be the set of all pricing measures which are consistent with the observable call prices having maturity $t_i$.
  \end{definition}
  
In what follows we assume that there is a strike price $K>0$ such that calls have price zero for strikes larger or equal than $K$, i.e. the pricing measures are concentrated on the compact interval $[0,K]$. Moreover we assume that call prices $c_0^i>\ldots >c_{n_i}^i=0$ are available for strikes $0= k_0^i<\ldots <k_{n_i}^i$ where  by our assumption $c_{0}^i=1$. We denote by $S_i := \{k_0^i < \ldots <k_{n_i}^i \}, i=1,2$ the set of strike prices for which call prices are observable. We choose the  functions $C^*_\mu, C^*_\nu$  to be exactly the functions that result from interpolating between the observed call option prices $(c_j^1)$ and $(c_j^2)$ respectively (see Figure \ref{fig:cstar}). We concentrate the discussion in this section on $C^*_\mu$ and in order to ease notation write $k_j$ instead of $k_j^1$.

That is, for $C^*_\mu(k_j)=c_j, j=0,\ldots,n_1$ and  for $k \in [k_{j},k_{j+1}), j=0,\ldots, n,$ we define
  \[
    C^*_\mu(k) := \frac{k_{j+1}-k}{k_{j+1}-k_j}C^*_\mu(k_j) + \frac{k-k_j}{k_{j+1}-k_j}C^*_\mu(k_{j+1}).
  \]
  
     \begin{figure}[!htbp]\begin{center}
       \begin{tikzpicture}[scale=1]      
            \draw(0,0)--(9,0);
             \draw(1,5)--(2,4.1);
       \draw(2,4.1)--(3,3.4);
       \draw(3,3.4)--(4,2.8);
           \draw(4,2.8)--(5,2.4);  
           \draw(5,2.4)--(6,2.2);
     \draw(6,2.2)--(7,2.1);   \draw(7,2.1)--(8,2.05);

    \draw(1,0.1)--(1,-0.1);
      \draw(2,0.1)--(2,-0.1);
        \draw(3,0.1)--(3,-0.1);
          \draw(4,0.1)--(4,-0.1);
            \draw(5,0.1)--(5,-0.1);
              \draw(6,0.1)--(6,-0.1);
                \draw(7,0.1)--(7,-0.1);

             \draw[dashed](1,5)--(1,0);
       \draw[dashed](2,4.1)--(2,0);
       \draw[dashed](3,3.4)--(3,0);
           \draw[dashed](4,2.8)--(4,0);  
           \draw[dashed](5,2.4)--(5,0);
     \draw[dashed](6,2.2)--(6,0);   \draw[dashed](7,2.1)--(7,0);
     \draw[dashed](8,2.05)--(8,0);
                
           \node  at (1,-0.5) {$k_0^i$};
               \node  at (2,-0.5) {$k_1^i$};
                 \node  at (3,-0.5) {$k_2^i$};
                   \node  at (4,-0.5) {$k_3^i$};
                     \node  at (5,-0.5) {$k_4^i$};
                       \node  at (6,-0.5) {$k_5^i$};
                         \node  at (7,-0.5) {$k_6^i$};
                           \node  at (8,-0.5) {$k_7^i$};   
         
          \draw(0,-0)--(0,5.1);           
               \node  at (-0.5,5) {$c_0^i$};
               \node  at (-0.5,4.1) {$c_1^i$};
                 \node  at (-0.5,3.4) {$c_2^i$};
                   \node  at (-0.5,2.8) {$c_3^i$};
                     \node  at (-0.5,2.4) {$\vdots$};
         
        \end{tikzpicture}
        \caption{$C_\mu^*, C_\nu^*$ interpolate the observed call prices. }
         \label{fig:cstar}
  \end{center}
 
    \end{figure}
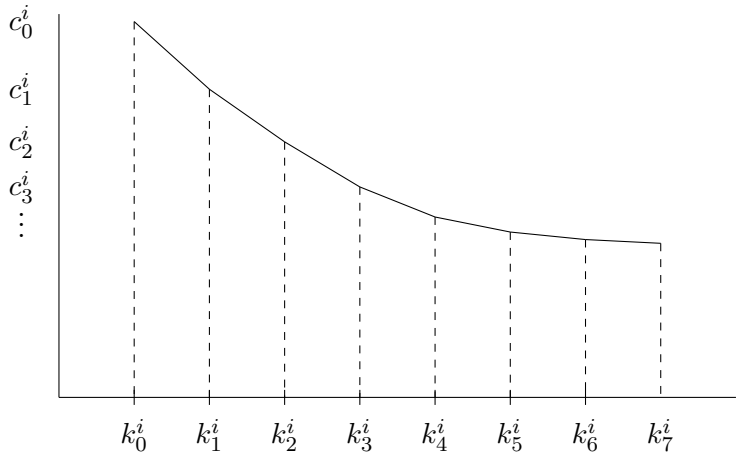

We assume that the observed prices are arbitrage-free and thus $C^*_\mu$ is a candidate function according to Lemma \ref{Lem Call Prop 1} (see Theorem 3.1 in \cite{dh07}).

\begin{lemma}\label{lem:mu*}
  The  measure $\mu^*$ consistent with $C^*_\mu$ is a discrete measure of the form 
  \begin{align*}
    \mu^* :=  \sum_{j=0}^{n}\omega_j\delta_{k_{j}} & := \sum_{j=0}^{n}\left[ \frac{C^*_\mu(k_{j+1})-C^*_\mu(k_{j})}{k_{j+1}-k_{j}} - \frac{C^*_\mu(k_{j})-C^*_\mu(k_{j-1})}{k_{j}-k_{j-1}}\right]\delta_{k_{j}}
  \end{align*}
  where we set $\frac{C^*_\mu(k_{n+1})-C^*_\mu(k_{n})}{k_{n+1}-k_{n}}= 0$ and $\frac{C^*_\mu(k_{0})- C^*_\mu(k_{-1})}{k_{0}-k_{-1}}= -1$ and $\delta_x$ is the Dirac measure on point $x$. 
\end{lemma}  

\begin{proof}
It is a routine calculation to see that $C^*_\mu(k)= \int_k^K (x-k)^+\mu^*(dx)$. A detailed calculation can be found in \cite{schm18} Appendix A.4. Note that this construction has also been used in \cite{cdcv08}.
\end{proof}

By construction the measure $\mu^*$ has a special property.

\begin{lemma}\label{lem:mu*extremal}
Suppose  that $\mu\in P_1$, i.e. $\mu$ is another probability measure consistent with the observable call prices in $P_1$.    Then $$ \mu\le_c \mu^*.$$ 
\end{lemma}

\begin{proof}
Let $C_{\mu}$ be consistent with $\mu$, i.e. it is by definition convex and passes through the points $(k_j,c_j).$ Moreover, $\int x\mu(dx)=\int x\mu^*(dx)=1$. But the convexity implies that $C_{\mu}\le C_{\mu^*}$ on $\R_+$. 
This in turn implies by Lemma \ref{lem:cx2} that $\mu\le_c \mu^*.$
\end{proof}

Hence with respect to  convex order $\mu^*$ is the maximal element of the set $P_1$. It is also possible to determine the Wasserstein distance of $\mu^*$ to any other element of $P_1$ explicitly. The Wasserstein distance between two probability measures is defined as follows. Let 
 \begin{align*}
    \Pi&(\mu,\mu^*) :=     \big\{ \pi \in P(\R^2)  : \mu(B)=\pi(B\times \R),\; \mu^*(C)= \pi(\R\times C),  B,C\in\mathcal{B}(\R)\big\}
  \end{align*}
  Then we define
\begin{definition}
The \textit{Wasserstein distance} of two probability measures $\mu,\mu^* \in {P}(\R)$ is given by
  \begin{equation}\label{Eq Def Wasserstein}
    W(\mu,\mu^*):= \inf_{\pi \in \Pi(\mu,\mu^*)} \int | x-y| \pi(dx,dy).
  \end{equation}
\end{definition}

\begin{remark}\label{rem:W}
\begin{itemize}
\item[(a)] If $F_{\mu}$ and $F_{\mu^*}$ are the cumulative distribution functions of   $\mu$ and $\mu^*$, it also holds that (see \cite{da}, Sec. 1)
 \[
        W(\mu,\mu^*) =\int_{-\infty}^{\infty} |F_{\mu}(t)-F_{\mu^*}(t)| d t.
    \]
\item[(b)]   There is also a dual representation of the Wasserstein distance:
 \[
        W(\mu,\mu^*) =\sup_{f\in C_1(\R)} \int f(x) (\mu-\mu^*)(dx),
    \]
    where $C_1(\R):= \{f:\R\to \R : f \mbox{ is Lipschitz-continuous with constant } 1\}$, see e.g. \cite{villani} Theorem 6.9.
\end{itemize}

\end{remark}    
Now we can show that

 \begin{theorem}\label{Thm Wasserstein Distance Estimates}
    Let $\mu \in P(\R_+)$ with $supp(\mu) \subset [0,K]$. Moreover choose $k_j = \frac{jK}{2^n}$, $j=0,\ldots, 2^n$,  $n \in \N$. Then we have
    \begin{align}\label{Thm Eq 1 Wasserstein Distance Estimates}
      W( \mu,\mu^*) = 2\cdot \sum_{j=0}^{2^n-1} \sup_{k\in [k_j,k_{j+1})} | C_{\mu^{*}}(k)-C_{\mu}(k)| \leq \frac{K}{2^n},
    \end{align}
    If we additionally assume that $C_{\mu} \in C^2(\R_{+})$, then, for any $n \in \N$, we have
    \begin{align}\label{Thm Eq 3 Wasserstein Distance Estimates}
      W( \mu,\mu^{*})  \leq \frac{T_{\mu}\cdot K^2 }{2^{n+1}},
       \end{align}
       where $T_\mu:= \sup\limits_{\kappa \in [0,K]}|C_{\mu}''(\kappa)|$.
  \end{theorem}    
  
A proof can be found in the appendix.  

\section{Bounds for Exotic Options}\label{sec:bounds}
As explained in the beginning we  consider a  risky asset at two time points $(S_{t_1},S_{t_2})=(X,Y).$ We are interested in finding upper bounds for prices of contingent claims of the form $c(X,Y),$ given a finite number of call prices on the same stock for both maturities $t_1$ and $t_2$.  A typical example for such a contingent claim would be an Asian option with payoff $c(X,Y)=(\frac12(X+Y)-K)^+.$
In what follows we denote by $(C_\mu,\mu)$ a consistent pair of price function and measure for $X$ and $(C_\nu,\nu)$  for $Y$. Since the general pricing theory implies that $(X,Y)$ is a martingale, i.e. $\Eop[Y|X]=X$ we must have by Lemma \ref{lem:cx1} that $\mu\le_c\nu$. We will restrict our consideration now to finite discrete measures $\mu$ and $\nu$ since every measure $\mu$ can be approximated arbitrary well by a sequence of finite measures. Moreover we assume that $S_1 \supset S_2$
i.e. if a call price is observable for strike $k_j^2$ for maturity $t_2$ then  a call price for the same strike is  also observable for maturity $t_1$.  We assume here that the observable call prices  are such that an arbitrage-free pricing model exists. Conditions for this are e.g.\ given in \cite{dh07}, Theorem 3.1. In case $S_1=S_2$ these boil down to $C^*_\mu$ and $C_\nu^*$ being candidate functions and $C_\mu^* \le C_\nu^*$ (see Corollary 4.1 in \cite{dh07}).

In what follows suppose that  $P^2$ is the space of all finite, discrete probability measures on $\R_+^2$. We will denote by $\Q\in P^2$ the joint distribution of $(X,Y)$. Then we define for two (discrete) probability measures $\mu \in P_1,\nu\in P_2$ with $\mu\le_c\nu$
\begin{eqnarray}
\mathcal{M}(\mu,\nu) \!\!\!\!&:=& \!\!\!\!\{\Q\!\in\! P^2\! : \Q(\cdot,\R_+\!)\!= \!\mu, \; \Q(\R_+,\cdot)\!= \!\nu, \int y \Q(x,dy) =x,   \Q\!\!-\!\!\text{a.s.} \},\\
\mathcal{M}(\mu,\cdot) \!\!\!\!&:=& \!\!\!\!\{\Q\in \mathcal{M}(\mu,\nu) :  \nu\in P_{2}, \mu\le_c\nu\}, \quad\mu\in P_1\\
\mathcal{M}(\cdot,\nu) \!\!\!\!&:=&\!\!\!\! \{\Q\in \mathcal{M}(\mu,\nu) : \mu \in P_{1}, \mu\le_c\nu\},\quad \nu\in P_2\\
\mathcal{M} \!\!\!\! &:=&\!\!\!\! \{\Q\in \mathcal{M}(\mu,\nu) : \mu \in P_{1}, \nu\in P_{2}, \mu\le_c\nu \}.
\end{eqnarray}
Note that by Lemma \ref{lem:cx1} the set $\mathcal{M}(\mu,\nu)$ is not empty if and only if $\mu\le_c\nu.$ 
The set $\mathcal{M}(\mu,\nu)$ is the set of all so-called martingale transports given the knowledge of the marginal distributions $\mu$ and $\nu$, i.e. the set of all potential pricing measures when we know the marginal distributions. However this is not the case in reality. Thus we consider the set $\mathcal{M}$ which consists of all potential pricing measures which are consistent with the finitely many observable call prices. The next result can easily be checked:

\begin{lemma}\label{lem:transcond}
Let $\Q\in P^2$ be a probability  measure and $\mu\le_c\nu$. Then $\Q\in \mathcal{M}(\mu,\nu)$  if and only if
\begin{itemize}
\item[(a)] $  \sum_{x} \Q(x,y)= \nu(y), \; \mbox{ for all } y\in supp(\nu),$
\item[(b)] $ \sum_{y} \Q(x,y) = \mu(x) \; \mbox{ for all } x\in supp(\mu),$
\item[(c)] $  \sum_{y} \Q(x,y)y = x\mu(x) \; \mbox{ for all } x\in supp(\mu),$
\end{itemize}
where $supp(\mu)$ and $supp(\nu)$ are the supports of measures $\mu$ and $\nu$ respectively.
\end{lemma}

The first two equations guarantee the correct marginals, equation (c) expresses the martingale property. 

We are now interested in finding
\begin{equation}
\sup_{\Q\in \mathcal{M} } \Eop_\Q [c(X,Y)]
\end{equation}
an upper  price bound for the exotic option with payoff $c$, consistent with the given finite number of call price observations. Unlike in classical transportation problems we cannot separate the optimization problem into finding the best marginals and then the best copula, because in optimal martingale transport the optimal transportation plan depends crucially on the margins. However, under some assumptions we can  nevertheless solve the problem. 

\begin{theorem}\label{theo:upperb}
If $\mu\in P_{1}$ is an arbitrary, fixed (discrete) probability measure and $y\mapsto c(\cdot,y)$ is convex, then
\begin{equation}
\sup_{\Q\in \mathcal{M}(\mu,\cdot)  } \Eop_\Q [c(X,Y)] = \sup_{\Q\in \mathcal{M}(\mu,\nu^*)  } \Eop_\Q [c(X,Y)] 
\end{equation} 
\end{theorem}

\begin{proof}
Let $\nu\in P_{2} $ be an arbitrary, consistent, discrete probability measure with $\mu\le_c\nu$. By Lemma \ref{lem:mu*extremal} we know that $\nu\le_c\nu^*$. Choose $\bar{\Q}\in  \mathcal{M}(\mu,\nu)$ and $\tilde{\Q}\in \mathcal{M}(\nu,\nu^*)$ arbitrary.  Then we define a new probability measure $\Q\in P^2$ as follows:
$$ \Q(x,z) := \sum_{y} \frac{\bar{\Q}(x,y)}{\nu(y)} \tilde{\Q}(y,z), \quad x\in supp(\mu), z\in  supp(\nu^*).$$
We show that  $\Q\in \mathcal{M}(\mu,\nu^*)$. First we check that $\Q$ has the correct marginal distributions. For this we use that $\tilde{\Q}$ and $\bar{\Q}$ satisfy the condition of Lemma  \ref{lem:transcond}:
 \begin{eqnarray*}
 \sum_z \Q(x,z) &=& \sum_z \sum_y \frac{\bar{\Q}(x,y)}{\nu(y)} \tilde{\Q}(y,z) = \sum_y \frac{\bar{\Q}(x,y)}{\nu(y)} \sum_z \tilde{\Q}(y,z)=  \sum_y \bar{\Q}(x,y)=\mu(x),\\
 \sum_x \Q(x,z) &=& \sum_x \sum_y \frac{\bar{\Q}(x,y)}{\nu(y)} \tilde{\Q}(y,z) = \sum_y \tilde{\Q}(y,z) = \nu^*(z).
 \end{eqnarray*} 
Second the martingale property is also satisfied:
\begin{eqnarray*}
\sum_z \Q(x,z) z &=& \sum_z \sum_y \frac{\bar{\Q}(x,y)}{\nu(y)} \tilde{\Q}(y,z) z= 
\sum_y \frac{\bar{\Q}(x,y)}{\nu(y)} \sum_z \tilde{\Q}(y,z) z\\
&=& \sum_y \bar{\Q}(x,y) y= x\mu(x).
\end{eqnarray*}
Finally the expected value of the payoff $c(X,Y)$ under $\Q$ is larger than under $\bar{\Q}$, because
\begin{eqnarray*}
\sum_x\sum_z \Q(x,z) c(x,z) &=& \sum_x\sum_z \sum_y \frac{\bar{\Q}(x,y)}{\nu(y)} \tilde{\Q}(y,z) c(x,z)\\
&=& \sum_x \sum_y \bar{\Q}(x,y) \sum_z \frac{\tilde{\Q}(y,z)}{\nu(y)}c(x,z)\\
&\ge&  \sum_x \sum_y \bar{\Q}(x,y) c\Big(x,  \sum_z \frac{\tilde{\Q}(y,z)}{\nu(y)}z\Big) \\
&=&  \sum_x \sum_y \bar{\Q}(x,y) c(x,y).
\end{eqnarray*}
Thus $\Eop_\Q [c(X,Y)]\ge \Eop_{\bar{\Q}}[c(X,Y)]$ for all $\nu\in P_2 $  with $\mu\le_c\nu$.
This proves the statement.
\end{proof}

Theorem \ref{theo:upperb} implies that if the payoff function $c$ is convex in the second component, then  we can fix $\nu^*$ as worst margin for the second component of the pricing measure, irrespective of the first component.

\begin{corollary} 
If  $y\mapsto c(\cdot,y)$ is convex, then
\begin{equation}
\sup_{\Q\in \mathcal{M}  } \Eop_\Q [c(X,Y)] = \sup_{\Q\in \mathcal{M}(\cdot,\nu^*)  } \Eop_\Q [c(X,Y)] 
\end{equation} 

\end{corollary}

Discussing the first component is more involved. We  need the following definition. 

\begin{definition}
Let $\Q$ be a measure on $\R_+^2$. 
\begin{itemize}
\item[a)] $\Q$ is called {\em direct transport} if  $\Q=\omega\delta_{(x,x)}$ for some $\omega\in(0,1)$ and $x\in\R$.
\item[b)] $\Q$ is called  {\em two-way transport} if $\Q= \vartheta_1 \delta_{(x,y_1)}+\vartheta_2\delta_{(x,y_2)}$ for some $\vartheta_1,\vartheta_2\in(0,1)$ and $x,y_1,y_2\in\R$ with $ y_2<x<y_1$ and with the property
$$(\vartheta_1+\vartheta_2) x = \vartheta_1 y_1+\vartheta_2 y_2.$$
\end{itemize}
\end{definition}

Note that if $\Q\in P^2$ in part b), then $\Q\in \mathcal{M} (\delta_x,\vartheta_1 \delta_{y_1}+\vartheta_2\delta_{y_2})$. The next lemma shows that every $\Q\in\mathcal{M}(\mu,\nu)$ can be decomposed into a finite number of direct and two-way transports.

\begin{lemma}\label{lem:twoway}
Suppose $\mu,\nu$ are two discrete probability measures. Then  $\Q\in\mathcal{M}(\mu,\nu)$ if and only if $\Q$ has margins $\mu$ and $\nu$ and can be written as a finite sum of direct and two-way transports.
\end{lemma}

\begin{proof}
Suppose $\Q\in\mathcal{M}(\mu,\nu)$. Consider the measure $\tilde{Q}$ which we obtain when we fix $x\in supp(\mu)$: $\tilde{Q} = \sum_{y\in supp(\nu)} \delta_{(x,y)} \Q(x,y)$. The margins $\tilde{\mu}, \tilde{\nu}$ of $\tilde{Q}$ are by Lemma \ref{lem:transcond} in convex order and $\tilde{Q}$ is the unique element in $\mathcal{M}(\tilde{\mu},\tilde{\nu})$. It can e.g.\ be constructed using the primal algorithm  in Sec. 5 of \cite{bschm18}. From this algorithm we see that $\tilde{\Q}$ can be decomposed into a finite number of direct and two-way transports. But one could also start constructing the two-way transports from the extreme points in $\tilde{\nu}$.

 This can be done with all $x\in supp(\mu)$ and the decomposition of $\Q$ follows. 

Contrary suppose $\Q = \sum_{j=1}^n \Q_j$ where $\Q_j$ are direct or two-way transports and $\Q$ has the right margins. It is left to show that $\Q\in \mathcal{M}(\mu,\nu)$, i.e. we have to show the martingale property. But this is true if it separately holds for all $\Q_j$ which is the case by definition.
\end{proof}

The next lemma will be crucial for our main result: 

\begin{lemma}\label{lem:mutilde}
Suppose $\mu\in P_1$ and $\mu$ has a positive mass $\kappa>0$ on a point $\tilde{x} \notin S_1.$ Denote by $x_2 := \max \{x \in S_1\cup supp(\mu): x < \tilde{x}\}$ and by $x_1 := \min \{x \in S_1\cup supp(\mu): x > \tilde{x}\}.$ Then
$$\tilde{\mu}:= \mu - \kappa\delta_{\tilde{x}}+\kappa\frac{x_1-\tilde{x}}{x_1-x_2}\delta_{x_2}+\kappa\frac{\tilde{x}-x_2}{x_1-x_2}\delta_{x_1}$$
satisfies $\tilde{\mu}\in P_1$,  $\mu\le_c \tilde{\mu}$ and $\tilde{\mu}$ can be obtained from $\mu$ by a two-way transport.
\end{lemma}

\begin{proof}
The situation is best illustrated with Figure \ref{fig:twowayle} which shows a part of the functions $C_\mu$ and $C_{\tilde{\mu}}$.
   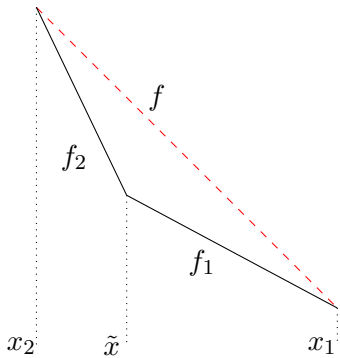
\begin{figure}[!htbp]\begin{center}
       \begin{tikzpicture}[scale=1]

               \draw[dashed,red](0,4)--(4,0);
            \draw(0,4)--(1.2,1.5);
             \draw(1.2,1.5)--(4,0);
       
         \draw[dotted](0,4)--(0,-0.5);
        \node  at (-0.2,-0.5) {$x_2$};
         \draw[dotted](1.2,1.5)--(1.2,-0.5);
        \node  at (1.0,-0.5) {$\tilde{x}$};
           \draw[dotted](4,0)--(4,-0.5);
        \node  at (3.8,-0.5) {${x_1}$};
        
           \node  at (0.5,2) {$f_2$};
          \node at (2.2,0.6) {$f_1$};
           \node at (1.6,2.8) {$f$};                  

        \end{tikzpicture}
        \caption{$C_\mu$ (solid line) and $C_{\tilde{\mu}}$ (dashed line) on the interval $[x_2,x_1]$  in the situation of Lemma \ref{lem:mutilde}.}
          \label{fig:twowayle}
  \end{center}
 
    \end{figure}
    
    Recall that $C_\mu(k) := \sum_{x\in supp(\mu)} (x-k)^+ \mu(x)$, i.e.\ the function is piecewise linear and has kinks at points which have positive mass.
The (black) solid line belongs to $C_\mu$. We will show that the dashed (red) line belongs to $C_{\tilde{\mu}}$. In order to do so let us write down the three linear functions $f_1,f_2,f$ in Figure \ref{fig:twowayle}.

For $x\in [x_2,\tilde{x}]$ we have $f_2(x) = sx+c_2$ where $s<0$ is the slope and $c_2\in\R$. 

For $x\in [\tilde{x},x_1]$ we have $f_1(x) = (s+\kappa)x+c_1$. Both functions coincide at $\tilde{x}$, thus $c_2 = \kappa\tilde{x}+c_1.$ The functional form of the dashed (red) line is
$$ f(x) = \frac{f_1(x_1)-f_2(x_2)}{x_1-x_2}x+\tilde{c}.$$
Since $f_2(x_2)= sx_2+\kappa \tilde{x}+c_1$ and $f_1(x_1)=(s+\kappa)x_1+c_1$ we obtain
$$  \frac{f_1(x_1)-f_2(x_2)}{x_1-x_2} = \frac{1}{x_1-x_2}\big( s(x_1-x_2)+\kappa x_1-\kappa \tilde{x}\big)= s+ \kappa \frac{x_1-\tilde{x}}{x_1-x_2}.$$
Thus we get $$f(x)=\Big( s+ \kappa \frac{x_1-\tilde{x}}{x_1-x_2}\Big) x + \tilde{c}.$$
Since $f(x_2)=f_2(x_2)$ and $f(x_1)=f_1(x_1)$ we obtain $\tilde{c}= c_1+  \kappa x_1\frac{\tilde{x}-{x_2}}{x_1-x_2}$. The additional mass $\tilde{\mu}$ places on $x_2$ is the difference of the slopes of $f$ and $f_2$ at $x_2+$ which is $$\theta:= \kappa\frac{x_1-\tilde{x}}{x_1-x_2}\in (0,\kappa)$$
and the additional mass $\tilde{\mu}$ places on $x_1$ is the difference of the slope of $f$ and $f_1$ at $x_1-$ which is given by
$$\kappa-\theta= \kappa\frac{\tilde{x}-x_2}{x_1-x_2}.$$
$\tilde{\mu}$ has no mass at $\tilde{x}$ any more since there is no kink here. Hence $f$ really corresponds to $\tilde{\mu}$ as defined in the statement. That $\tilde{\mu} \ge_c \mu$ is clear from the construction and Lemma \ref{lem:cx2}.  Moreover, $C_{\tilde{\mu}}$ is still the same as $C_\mu$ on $(-\infty,x_2]$ and on $[x_1,\infty)$ thus it satisfies $\tilde{\mu}\in P_1$. Finally note that $\delta_{(\tilde{x},x_1)} \kappa\frac{\tilde{x}-x_2}{x_1-x_2}+\delta_{(\tilde{x},x_2)} \kappa\frac{x_1-\tilde{x}}{x_1-x_2}$ is a two-way transport.
\end{proof}



In order to deal with the second worst case margin for price bounds we need a special property of the payoff $c$.

\begin{definition}
A function $c:\R^2\to\R$ is called {\em directionally convex} if $c$ is convex in both components and $c$ is {\em supermodular}, i.e. for all $x,y\in\R^2$
$$ f(x)+f(y)\le f(x\wedge y)+f(x\vee y),$$
where $x\wedge y=(\min\{x_1,y_1\},\min\{x_2,y_2\})$ and $x\vee y=(\max\{x_1,y_1\},\max\{x_2,y_2\})$.
\end{definition}

\begin{remark}
For equivalent definitions of the directionally convex property, see \cite{ms} Theorem 3.12.2. In particular if $f\in C^2$, then $f$ is directionally convex if and only if
$$ \frac{\partial^2}{\partial x_i\partial x_j} f(x)\ge 0, \quad \mbox{ for all } x\in\R^2, i,j=1,2.$$

\end{remark}

Next we can show

\begin{theorem}\label{th:main}
If  $ c$ is directionally convex, then
\begin{equation}\label{eq:prob}
\sup_{\Q\in \mathcal{M}  } \Eop_\Q [c(X,Y)] = \sup_{\Q\in \mathcal{M}(\mu^*,\nu^*)  } \Eop_\Q [c(X,Y)] 
\end{equation} 
\end{theorem}

\begin{proof}
We know already that we can fix the upper margin to be $\nu^*$.  Now let $\mu\in P_1$ be arbitrary and $\mu^*$ the convex upper bound construction from Lemma \ref{lem:mu*extremal}.   Let  $\Q\in \mathcal{M}(\mu,\nu^*)$ be an arbitrary martingale transport. We will show that as long as $\mu$ possess mass on atoms $\tilde{x}$ which are no strike prices (like in Lemma \ref{lem:mutilde}),  the expectation $\Eop_\Q [c(X,Y)]$ can be increased. 

Suppose $\mu$ has a positive mass on a point $\tilde{x}\notin S_1$. By Lemma \ref{lem:mutilde} there is a two-way transport from $\tilde{x}$ which shifts mass to neighbouring points $x_1$ and $x_2$ which are either strike prices or other atoms of $\mu$. By Lemma \ref{lem:twoway} there is also a two-way transport from $\tilde{x}$ to $y_1, y_2\in supp(\nu^*)$ in $\Q$. 
In both transports a certain mass from $\tilde{x}$ is transported and we consider the smaller of the two masses. I.e. in the other transport we consider the respective part of the mass. Suppose this mass is $\kappa>0$. More precisely we consider now the following two transports. The first between $\mu$ and $\nu^*$: $ \alpha\delta_{(\tilde{x},y_1)}+(\kappa-\alpha)\delta_{(\tilde{x},y_2)}$ with $0\le \alpha\le \kappa\le 1$ and the second transport between $\mu$ and $\tilde{\mu}$ (see Figure \ref{fig:zweiTP}):  $\theta\delta_{(\tilde{x},x_1)}+(\kappa-\theta)\delta_{(\tilde{x},x_2)}$ with $0\le \theta\le \kappa\le 1.$ Note that by definition and our assumption on the strike prices  we have that $y_2\le x_2 <x_1\le y_1$. 

   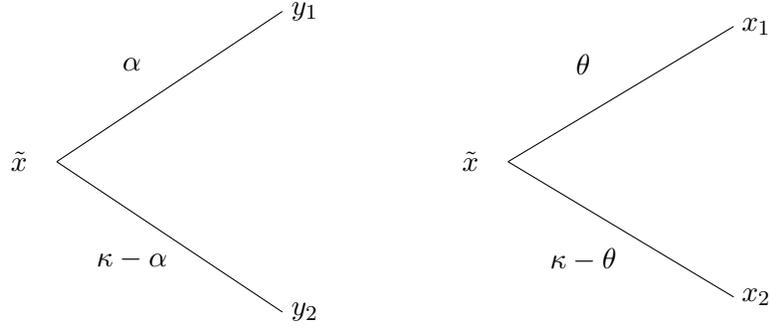
\begin{figure}[!htbp]\begin{center}

       \begin{tikzpicture}[scale=1]

               \draw(0,0)--(3,2);
            \draw(0,0)--(3,-2);
       
         \node  at (-0.5,0) {$\tilde{x}$};
          \node at (3.3,2) {$y_1$};
           \node at (3.3,-2) {$y_2$};                  
          \node  at (1,1.3) {$\alpha$};
            \node at (1,-1.3) {$\kappa-\alpha$};
        
            \draw(6,0)--(9,1.8);
            \draw(6,0)--(9,-1.8);
       
         \node  at (5.5,0) {$\tilde{x}$};
          \node at (9.3,1.8) {$x_1$};
           \node at (9.3,-1.8) {$x_2$};                  
          \node  at (7,1.3) {$\theta$};
            \node at (7,-1.3) {$\kappa-\theta$};

        \end{tikzpicture}
        \caption{Left: Two-way transport from $\mu$ to $\nu^*$. Right: Two-way transport from $\mu$ to $\tilde{\mu}$}
        \label{fig:zweiTP}
  \end{center}
    \end{figure}
 
In particular by the definition of a two-way transport it holds that
\begin{itemize}
\item[(i)] $\alpha y_1+(\kappa-\alpha)y_2= \kappa \tilde{x}$.
\item[(ii)] $\theta x_1+(\kappa-\theta) x_2= \kappa \tilde{x}$.
\end{itemize}

Now we will replace the transport $\Q$ from $\mu$ to $\nu^*$ by one which yields a higher expectation and where the first margin is  $\tilde{\mu}$. We consider the following transport (see Figure \ref{fig:construction}):
$$ \Q^* := \Q - \alpha \delta_{(\tilde{x},y_1)}-(\kappa-\alpha)\delta_{(\tilde{x},y_2)}+\vartheta_1 \delta_{(x_1,y_1)}+(\theta-\vartheta_1) \delta_{(x_1,y_2)} +\vartheta_2 \delta_{(x_2,y_1)}+(\kappa-\theta-\vartheta_2) \delta_{(x_2,y_2)} $$
where it holds that
\begin{itemize}
\item[(iii)] $\vartheta_1+\vartheta_2=\alpha$.
\item[(iv)] $\theta x_1 = \vartheta_1 y_1+(\theta-\vartheta_1) y_2$.
\item[(v)] $\theta-\vartheta_1+\kappa-\theta-\vartheta_2=\kappa-\alpha$.
\item[(vi)] $(\kappa-\theta) x_2 = \vartheta_2 y_1+(\kappa-\theta-\vartheta_2)y_2$.
\end{itemize}
It is not difficult to see that (v) follows from (iii) and that (vi) follows from (ii), (iii) and (iv).

   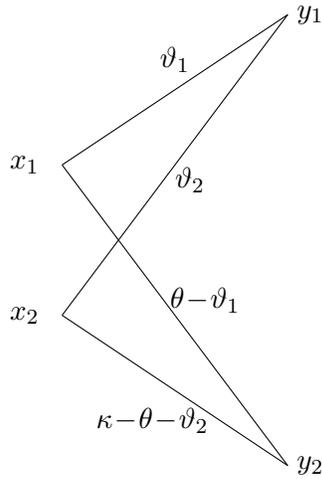
\begin{figure}[!htbp]\begin{center}

       \begin{tikzpicture}[scale=1]

               \draw(0,1)--(3,3);
            \draw(0,-1)--(3,3);
                \draw(0,1)--(3,-3);
            \draw(0,-1)--(3,-3);
          \node  at (-0.5,1) {$x_1$};
           \node  at (-0.5,-1) {$x_2$};
     
             \node at (3.3,3) {$y_1$};
           \node at (3.3,-3) {$y_2$};                  
     
          \node at (1.5,2.4) {$\vartheta_1$};
          \node at (1.2,-2.4) {$\kappa\!-\!\theta\!-\!\vartheta_2$};
       \node at (1.7,0.8) {$\vartheta_2$};
         \node at (1.9,-0.8) {$\theta\!-\!\vartheta_1$};

        \end{tikzpicture}
        \caption{Constructed transport in $\Q^*$.}
        \label{fig:construction}
  \end{center}
    \end{figure}

We claim now that the expected reward from this new transport is larger, i.e. we show that
\begin{equation}
c(\tilde{x},y_1) \alpha + c(\tilde{x},y_2)(\kappa-\alpha) \le c(x_1,y_1) \vartheta_1 + c(x_1,y_2)(\theta-\vartheta_1) + c(x_2,y_1) \vartheta_2 +c(x_2,y_2)(\kappa-\theta-\vartheta_2).\end{equation} 
In order to do this we first obtain from (ii) that
$$\frac{\theta}{\kappa} x_1+\frac{\kappa-\theta}{\kappa} x_2=  \tilde{x}$$
and thus by convexity of $c$ in the first component:
\begin{eqnarray*}
c(\tilde{x},y_1) \alpha + c(\tilde{x},y_2)(\kappa-\alpha) &\le& c(x_1,y_1) \frac{\theta}{\kappa} \alpha +c(x_1,y_2)  \frac{\theta}{\kappa} (\kappa-\alpha) \\
&&+ c(x_2,y_1)  \frac{\kappa-\theta}{\kappa}  \alpha + c(x_2,y_2) \frac{\kappa-\theta}{\kappa}(\kappa-\alpha).\end{eqnarray*} 
From (iii) and (iv) we get that
\begin{equation}
\vartheta_1 = \frac{\theta(x_1-y_2)}{y_1-y_2}, \quad \vartheta_2 = \alpha-\frac{\theta(x_1-y_2)}{y_1-y_2}.
\end{equation}
Thus it remains to show 
\begin{eqnarray*}
&&c(x_1,y_1) \frac{\theta}{\kappa} \alpha +c(x_1,y_2)  \frac{\theta}{\kappa} (\kappa-\alpha) + c(x_2,y_1)  \frac{\kappa-\theta}{\kappa}  \alpha + c(x_2,y_2) \frac{\kappa-\theta}{\kappa}(\kappa-\alpha)\\
&\le& c(x_1,y_1) \vartheta_1 + c(x_1,y_2)(\theta-\vartheta_1) + c(x_2,y_1) \vartheta_2 +c(x_2,y_2)(\kappa-\theta-\vartheta_2)\\
&=& c(x_1,y_1) \frac{\theta(x_1-y_2)}{y_1-y_2} + c(x_1,y_2)\Big(\theta-\frac{\theta(x_1-y_2)}{y_1-y_2}\Big) + c(x_2,y_1) \Big(\alpha-\frac{\theta(x_1-y_2)}{y_1-y_2}\Big)\\
&& +c(x_2,y_2)\Big(\kappa-\theta- \alpha+\frac{\theta(x_1-y_2)}{y_1-y_2}\Big).
\end{eqnarray*}
Rearranging the inequality is equivalent to
\begin{eqnarray*}
0&\le &c(x_1,y_1) \Big( \frac{\theta(x_1-y_2)}{y_1-y_2}-\frac{\theta}{\kappa} \alpha \Big) + c(x_1,y_2) \Big(\theta-\frac{\theta(x_1-y_2)}{y_1-y_2} -\frac{\theta}{\kappa} (\kappa- \alpha)   \Big)\\
&& + c(x_2,y_1)  \Big( \alpha-\frac{\theta(x_1-y_2)}{y_1-y_2}    - \frac{\kappa-\theta}{\kappa}\alpha\Big)+ c(x_2,y_2) \Big(\kappa-\theta- \alpha+\frac{\theta(x_1-y_2)}{y_1-y_2}- \frac{\kappa-\theta}{\kappa} (\kappa-\alpha)\Big).
\end{eqnarray*}
Simplifying the terms yields 
 \begin{eqnarray*}
0&\le &c(x_1,y_1) \Big( (x_1-y_2)-\frac{ \alpha}{\kappa} (y_1-y_2) \Big) - c(x_1,y_2) \Big((x_1-y_2)-\frac{ \alpha}{\kappa} (y_1-y_2)  \Big)\\
&& - c(x_2,y_1)  \Big((x_1-y_2)-\frac{ \alpha}{\kappa} (y_1-y_2) \Big)+ c(x_2,y_2) \Big((x_1-y_2)-\frac{ \alpha}{\kappa} (y_1-y_2)\Big).
\end{eqnarray*} 
Since the expression in brackets is positive if and only if  $x_1>\tilde{x}$ which is  true by assumption this is equal to
\begin{eqnarray*}
0\le c(x_1,y_1) - c(x_1,y_2) - c(x_2,y_1)  + c(x_2,y_2)
\end{eqnarray*} 
which is exactly the supermodularity from our assumption. Thus, when we replace $\Q$ by $\Q^*$ we have by definition a new martingale transport from $\tilde{\mu} :=\mu-\kappa \delta_x + \theta\delta_{x_1} +(\kappa-\theta) \delta_{x_2}$ to $\nu^*$ which still satisfies $\tilde{\mu}\le_c \mu^*$ but yields a higher expectation. This procedure can be repeated until all mass is concentrated on the strike prices. But this yields a unique measure which is $\mu^*$ and the statement is shown.
\end{proof}

\begin{example}
An important example for a directionally convex payoff function is the payoff of an Asian option $c(x,y)=(\frac12(x+y)-K)^+, x,y\ge0 $. It is obviously convex in both components and supermodularity follows from the fact that if $f:\R^2 \to \R$ is increasing and supermodular, then $\max\{f,c\}$ for all $c\in\R$ is supermodular (see e.g. \cite{b97}, Lemma 2.1 a)). For further constructions of supermodular functions, see Lemma 2.1 in \cite{b97}.
\end{example}

\begin{remark}
\begin{itemize}
\item[(a)] Suppose that \eqref{eq:prob} is solved by $\Q^*$ which has not necessarily discrete margins  $\mu$ and $\nu$. Then it is possible to approximate $\mu$ and $\nu$ by discrete measures $\mu^d \in P_1$ and $\nu^d\in P_2$ such that for the optimal martingale transport $\Q^d$
$$ \Big| \Eop_{\Q^d}[c(X,Y)]-\Eop_{\Q^*}[c(X,Y)]\Big|\le \varepsilon$$
for arbitrary $\varepsilon>0$. This follows since $(\mu,\nu) \mapsto \Eop_{\Q_l(\mu,\nu)}[c(X,Y)]$  is continuous with respect to weak convergence where $\Q_l$ is the left-monotone martingale transport (see \cite{J16}). It immediately implies  that Theorem \ref{th:main} is also true when we allow to optimize over all measures $\mu,\nu$ (not only discrete ones) which are consistent with observable call prices.
\item[(b)] Note that the method we use here can only be applied to upper bounds since there is no minimal element with respect to the convex order in the sets $P_i$.
\item[(c)] Once having the optimal transport for the upper bound, a super hedging strategy can be obtained from the dual problem (see e.g. \cite{blp13} for details).
\end{itemize}

\end{remark}

\section{Convergence of Price Bounds}
Of course we expect that the more call prices we observe, the closer the upper bound for the price of the exotic option is to the true upper bound which we could compute if we would know the true pricing measures for the asset. Indeed under the assumption that the strikes of observed call prices are given in a certain way,  it is not only possible to show the convergence of the upper price bound to the true upper price bound but also the speed of convergence can be estimated. We restrict here to the case that both margins have a compact support $[0,K]$. The situation with unbounded support is more complicated, see e.g. \cite{go17}. Moreover, we assume that in model $n$ call prices for strikes $k_{j}^{n} :=  \frac{j}{2^n}K, \;j=0,\ldots,2^n$ are observable for both margins at $t_1$ and $t_2$. We suppose that $\mu$ and $\nu$ are the true marginal distributions of $S_{t_1}$ and $S_{t_2}$ respectively. In the case of a finite number of observable call prices we take as marginal distribution the construction in Lemma \ref{lem:mu*extremal} which yields $\mu^*_n$ and $\nu^*_n$ irrespective of whether or not this really yields the upper bound (for $c$ directionally convex this is the case by Theorem \ref{theo:upperb}).   The lower index $n$ in the notation $\mu^*_n$ refers to the number of observable call prices. In order to establish convergence and the speed of convergence we need some more properties of $c$. More precisely we obtain:

  \begin{theorem}\label{Thm Convergence Speed General}
    Let $\mu\le_c\nu$ with $supp(\mu), supp(\nu)  \subset [0,K]$. Let $c:\R_{+}^2 \to \R$ be a Lipschitz continuous payoff function such that $c_{yy}$ exists. We denote by $\hat \Lambda$ the Lipschitz constant of $c$ and assume $\max\{\hat\Lambda, \sup_{(x,y) \in [0,K]^2} |c_{yy}(x,y)|\}\leq \Lambda$. Then, for any $n \in \N$, we have
    \begin{equation}\label{eq:WMc}
        \left|\sup_{\Q \in \mathcal{M}(\mu_n^{*},\nu_n^{*})} \Eop_{\Q}\left[c(X,Y) \right] - \sup_{\Q \in \mathcal{M}(\mu,\nu)}\Eop_{\Q}\left[c(X,Y) \right]\right|
        \leq \frac{M_c}{2^{n}},
        \end{equation}
       where $M_c= 12 K\tilde\Lambda$ with $\tilde \Lambda= \Lambda \cdot \max\{K,1\}$. If we additionally suppose that $C_{\mu},C_{\nu} \in \mathcal{C}^2(\R_+)$, then, for any $n \in \N$, we have
    \begin{equation}\label{eq:WMd}
      \left|\sup_{\Q \in \mathcal{M}(\mu_n^{*},\nu_n^{*})} \Eop_{\Q}\left[c(X,Y) \right] - \sup_{\Q \in \mathcal{M}(\mu,\nu)}\Eop_{\Q}\left[c(X,Y) \right]\right| \leq \frac{M_d}{2^{n+1}},
       \end{equation} 
     where $M_d=(7 T_{\mu} + 5 T_\nu )K^2\tilde\Lambda$ with $T_\mu= \sup\limits_{\kappa \in [0,K]}|C_{\mu}''(\kappa)|$ and $T_{\nu}=\sup\limits_{\lambda \in [0,K]}|C_{\nu}''(\lambda)|$.
\end{theorem}

\begin{proof}
We may reformulate the difference in \eqref{eq:WMc} as the difference of the values of the dual problems. Note that $c$ is bounded on $[0,K]^2$.  Thus we obtain with Theorem 1.1 in \cite{lim}
     \begin{align*}
       \sup_{\Q \in \mathcal{M}(\mu_n^{*},\nu_n^{*})}\Eop_{\Q}\left[c(X,Y) \right] &= \inf_{(\varphi,\psi) \in \mathcal{D}}\left\{ \int_{\R_+}\varphi(x)\mu_n^{*}(d x) + \int_{\R_+}\psi(y)\nu_n^{*}(d y) \right\},\\ \sup_{\Q \in \mathcal{M}(\mu,\nu)}\Eop_{\Q}\left[c(X,Y) \right] &= \inf_{(\varphi,\psi) \in \mathcal{D}}\left\{ \int_{\R_+}\varphi(x)\mu(d x) + \int_{\R_+}\psi(y)\nu(d y) \right\},
     \end{align*}
where \begin{eqnarray*}
\mathcal{D} &:=& \{ (\varphi,\psi) : \varphi^+ \in L^1(\R,\mu), \psi^+\in L^1(\R,\nu), \mbox{ and for some } h\in L^\infty(\R), \\
&& \hspace*{1cm}\varphi(x)+\psi(y)+h(x)(y-x) \ge c(x,y), (x,y)\in \R^2\}
\end{eqnarray*}
    Now let us apply Theorem 2.4 and Remark 2.5 in \cite{lim}. For this purpose, we have to  prove that the conditions are satisfied. By assumption and by construction respectively, we have that ${\mu\leq_c \nu}$ and $\mu_n^{*} \leq_c \nu_n^{*}$ are compactly supported. The payoff function $c$ is Lipschitz continuous on $[0,K]\times [0,K]=\mathrm{conv}(supp(\nu))\times \mathrm{conv}(supp(\nu))$ with constant $\hat \Lambda \leq \Lambda$. It remains to show that there is a Lipschitz continuous function ${u:[0,K]=\mathrm{conv}(supp(\nu)) \to \R}$ such that $y \mapsto c(x,y)+u(y)$
    is concave on $[0,K]$ for $\mu$-almost every $x\in \R_+$. As $c_{yy}\leq \Lambda$ on $[0,K]^2$, it is clear that $u(y):= -\frac{\Lambda}{2}y^2$ is such a function with Lipschitz constant $\Lambda K$. We define $\tilde \Lambda:=\Lambda \cdot \max\left\{ K,1 \right\}$.

    Thus, by Theorem 2.4 in \cite{lim} there exist solutions $(\varphi^{*},\psi^{*})$ and $(\varphi_n^{*},\psi_n^{*})$ for the dual problems with respect to $(\mu,\nu)$ and $(\mu_n^{*},\nu_n^{*})$ respectively. 
 By Remark 2.5 in \cite{lim} $\varphi^{*} $ and $\varphi_{n}^{*}$ are Lipschitz continuous with constant $7\tilde \Lambda$, and $\psi^{*}$ and $\psi_n^{*}$ are Lipschitz continuous with constant $5 \tilde\Lambda$. Hence, we have 
    \begin{align*}
        &\sup_{\Q \in \mathcal{M}(\mu_n^{*},\nu_n^{*})}\Eop_{\Q}\left[c(X,Y) \right] - \sup_{\Q \in \mathcal{M}(\mu,\nu)}\Eop_{\Q}\left[c(X,Y) \right]        
      \\
      &=\inf_{(\varphi,\psi) \in \mathcal{D}}\left\{ \int_{\R_+}\varphi(x)\mu_n^{*}(d x) + \int_{\R_+}\psi(y)\nu_n^{*}(d y) \right\}      
     \\ & \quad - \inf_{(\varphi,\psi) \in \mathcal{D}}\left\{ \int_{\R_+}\varphi(x)\mu(d x) + \int_{\R_+}\psi(y)\nu(d y) \right\}
      \\
      &\leq \int_{\R_+}\varphi^{*}(x)\mu_n^{*}(d x) + \int_{\R_+}\psi^{*}(y)\nu_n^{*}(d y)
      \\
      &\quad  -\left( \int_{\R_+}\varphi^{*}(x)\mu(d x) + \int_{\R_+}\psi^{*}(y)\nu(d y) \right)
      \\
      &= \int_{\R_{+}}\varphi^{*}(x)\left( \mu_{n}^{*}-\mu \right)(d x) + \int_{\R_+}\psi^{*}(y)\left( \nu_{n}^{*}-\nu \right)(d y)
      \\
      &\leq 7\tilde\Lambda W( \mu,\mu_{n}^{*}) + 5\tilde\Lambda W(\nu, \nu_{n}^{*}),
    \end{align*}
    where in the last inequality we scale the integrands by their Lipschitz constants and then use the dual representation of the Wasserstein distance in Remark \ref{rem:W}. Completely analogous, but using $\varphi^{*}_n$ and $\psi^{*}_n$ in the first inequality instead of $\varphi^{*}$ and $\psi^{*}$, we obtain
    \[ \sup_{\Q \in \mathcal{M}(\mu,\nu)}\Eop_{\Q}\left[c(X,Y) \right]   -\sup_{\Q \in \mathcal{M}(\mu_n^{*},\nu_n^{*})}\Eop_{\Q}\left[c(X,Y) \right]         
        \leq 7\tilde\Lambda W( \mu,\mu_{n}^{*}) + 5\tilde\Lambda W( \nu,\nu_{n}^{*} ).
    \]
    Using the estimates in \eqref{Thm Eq 1 Wasserstein Distance Estimates} and \eqref{Thm Eq 3 Wasserstein Distance Estimates}, we have the claimed convergence speed estimates.
\end{proof}

\begin{remark}
Note that the rate of convergence in Theorem \ref{Thm Convergence Speed General} cannot be improved. We show this by an example in the appendix. 
\end{remark}

\section{Numerical Examples}\label{sec:numerics}
\subsection{An Upper Price Bound Example with Real Data}
We use the theory developed in this paper to compute an upper price bound on an Asian option on the SAP stock. We have got the  call price observations in Table \ref{tab:SAP} taken May, 29th, 2019 from www.boerse.de. Note that we are only able to observe bid and ask prices and took a reasonable price from this interval as the net price such that the properties in Lemma \ref{Lem Call Prop 1} are satisfied. Moreover, the interest rate is not zero as assumed in our analysis. Prices are quoted in euros. The stock price of SAP at this time was 110 euros. As different maturities we have taken June, 17th, 2019 and August, 12th, 2019. 

   \begin{table}[!htbp]
\begin{tabular}{|c|c|c|c|c|c|c|c|c|}
\hline
strike & 90 & 95 & 100 & 105& 110 & 115 & 120 & 125\\
\hline 
Call price with maturity 17.6. & 2.2825 & 1.78 & 1.265 & 0.78 & 0.345 & 0.06 & 0.025 & 0.01 \\
\hline
Call price with maturity 12.8. & 2.405 & 1.907 & 1.414 & 0.976 & 0.613 & 0.365 & 0.2005 & 0.11 \\
\hline
\end{tabular}\\[0.4cm]

  \caption{Observed call prices on May 29th, 2019 in euro.}
  \label{tab:SAP}
\end{table}

Figure \ref{fig:SAP} shows the interpolated call price function. Using Lemma \ref{lem:mu*extremal} we can derive $\mu^*$ and $\nu^*$. What we clearly see in the picture is that $C_{\mu^*}\le C_{\nu^*}$ hence $\mu^* \le_c \nu^*$ follows from Lemma \ref{lem:cx2}. Moreover, the two functions are candidate functions in the sense of Definition \ref{Def Candidate}. Thus Corollary 4.1 in \cite{dh07} implies that the market is free of arbitrage.

   \begin{figure}[!htbp]\begin{center}

       \begin{tikzpicture}[scale=1]

               \draw(0,2.8)--(1,2.28);
            \draw(1,2.28)--(2,1.78);
                \draw(2,1.78)--(3,1.265);
                        \draw(3,1.265)--(4,0.78);
           \draw(4,0.78)--(5,0.345);
                \draw(5,0.345)--(6,0.06);
                       \draw(6,0.06)--(7,0.025);
           \draw(7,0.025)--(8,0.01);

     \draw[dashed, red](0,2.91)--(1,2.405);
            \draw[dashed, red](1,2.405)--(2,1.907);
                \draw[dashed, red](2,1.907)--(3,1.414);
                    \draw[dashed, red](3,1.414)--(4,0.976);
            \draw[dashed, red](4,0.976)--(5,0.613);
                \draw[dashed, red](5,0.613)--(6,0.365);
               \draw[dashed, red](6,0.365)--(7,0.2005);
            \draw[dashed, red](7,0.2005)--(8,0.11);
      
           \draw(0,0)--(8,0);
 \draw(0,0.1)--(0,-0.1);
  \node at (0,-0.4) {$85$};
 \draw(1,0.1)--(1,-0.1);
  \node at (1,-0.4) {$90$};
 \draw(2,0.1)--(2,-0.1);
  \node at (2,-0.4) {$95$};
 \draw(3,0.1)--(3,-0.1);
  \node at (3,-0.4) {$100$};
 \draw(4,0.1)--(4,-0.1);
  \node at (4,-0.4) {$105$};
 \draw(5,0.1)--(5,-0.1);
  \node at (5,-0.4) {$110$};
 \draw(6,0.1)--(6,-0.1);
  \node at (6,-0.4) {$115$};
 \draw(7,0.1)--(7,-0.1);
  \node at (7,-0.4) {$120$};
 \draw(8,0.1)--(8,-0.1);
  \node at (8,-0.4) {$125$};
         
        \end{tikzpicture}
        \caption{Interpolated call prices on SAP stock with maturity 17.6.19 (solid) and 12.8.19 (dashed).} \label{fig:SAP}
  \end{center}
    \end{figure}
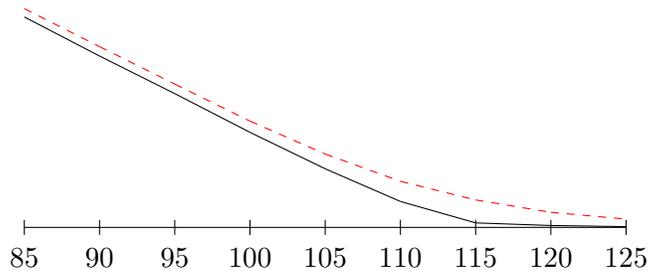
    
 We obtain 
 
    \begin{table}[!htbp]
\begin{tabular}{|c|c|c|c|c|c|c|c|c|}
\hline  
 atom & 90 & 95 & 100 & 105& 110 & 115 & 120 & 125\\ \hline
$\mu^*(\cdot)$ &  0.895 & 0.002 & 0.006 & 0.01 & 0.03 & 0.05 & 0.004 & 0.003\\
\hline
$\nu^*(\cdot)$ &  0.9004 & 0.001 & 0.011 & 0.015 & 0.023 & 0.0167 & 0.0148 & 0.0181\\
\hline
\end{tabular}\\[0.4cm]

  \caption{Density of marginal distributions.}
\end{table}    

Note that the expectation of $\mu^*$ and $\nu^*$ is $92.295$.  This should in theory be equal to the initial stock price of 110, as already mentioned we are not able to observe net prices.   We want to find an upper price bound for an Asian option on the SAP stock with payoff $c(x,y)=\Big(\frac12 (x+y)-K\Big)^+$ where $x$ is the stock price at 17.6.2019 and $y$ the stock price at 12.08.2019. When we choose $k=120$ we obtain the price bound $0.02357$. Note that we assume here that there is no discounting which is of course not true in reality, but the maturity of the option is not too far away. When we do not consider the martingale condition and just compute the upper price bound using the margins $\mu^*$ and $\nu^*$ we obtain for the same option the price $0.025$. Some earlier papers (e.g. \cite{cdcv08}) have considered bounds where the martingale condition is not respected. So in this example the additional condition does not improve the price bound much. But this is certainly due to the fact the two time points are pretty close together. 
    
\subsection{Numerical Convergence}
    Let us now discuss the convergence speed of the price bound approximation for different compactly supported theoretical marginals and payoff functions  numerically. Therefore, we calculate the approximating upper price bounds for several $n \in \N$ as well as the real upper price bounds as far as possible. Additionally, we calculate the corresponding normalized price bound differences $$d_n := 2^{2n} \Big(\sup_{\Q \in \mathcal{M}(\mu_n^{*},\nu_n^{*})} \Eop_{\Q}\left[c(X,Y) \right] - \sup_{\Q \in \mathcal{M}(\mu,\nu)}\Eop_{\Q}\left[c(X,Y) \right]\Big) .$$
In the case of uniform distributed margin, it is possible to compute the real upper price bound  explicitly (see e.g.  \cite{hlt16}) when $c$ has certain properties (like the Martingale Spence Mirrlees condition).

 In what follows $ \mathcal{U}[a,b]$ denotes the uniform distribution on interval $[a,b]$.

   \begin{enumerate}
     \item[1.] Let $\mu \sim \mathcal{U}[1,3], \nu \sim \mathcal{U}[0,4]$, i.e. $k=4$.  Here, we partition the support of $\nu$ in maximally $2048$ intervals, i.e. we have a difference of $\frac{1}{512}$ between two partition points. We consider different exotic options. 
     \begin{enumerate}
       \item $c\left(x,y\right)=xy^2$. Note that $c$ is directionally convex on $\R^2_+.$ The real upper price bound can be computed explicitly here, since the upper bound is attained at the left-monotone martingale transport (see e.g.  \cite{hlt16}). It is given by \begin{align*}
       \sup_{\Q \in \mathcal{M}\left(\mu,\nu\right)}\Eop_{\Q}\left[c\left(X,Y\right) \right] = \Eop_{\mu}\left[X\left(\frac{3}{4}\left(\frac{3}{2}X+\frac{1}{2}\right)^2+\frac{1}{4}\left(\frac{3}{2}-\frac{1}{2}X\right)^2\right)\right] = 12.5.
       \end{align*}
       For all $n=3,\ldots,11$, we calculate the approximate upper price bound $P(\mu_n^*, \nu_n^*)=  \sup_{\Q \in \mathcal{M}\left(\mu_n^*,\nu_n^*\right)}\Eop_{\Q}\left[c\left(X,Y\right) \right]$ and $d_n.$ This yields the results of Table \ref{Tab Uni 0134 xyy}.
       \begin{table}[!htbp]\scriptsize
         \centering \renewcommand{\arraystretch}{1.2}
           \begin{tabular}{|c|c|c|c|c|c|c|c|c|c|}
      \hline
      n &  3 & 4 & 5 & 6 & 7 & 8 & 9 & 10 & 11 \\
      \hline
      $P(\mu_n^*,\nu_n^*)$  & 12.808 & 12.57 & 12.517 & 12.504 & 12.501 & 12.5002 & 12.50006 & 12.50002 & 12.500004 \\ \hline
      $d_n$  & 19.7 & 17.9 & 17.1 & 16.6 & 16.4 & 16.35& 16.3 & 16.274 & 16.274 \\ \hline
           \end{tabular}
           \caption{Approximation results in the case 1.a)}
       \label{Tab Uni 0134 xyy}
       \end{table}

     \item $c\left(x,y\right)=\exp(x)\cdot y^2$. Note that $c$ is directionally convex on $\R^2_+.$ The real upper price bound is $ 61.8801.$ The approximate upper price bounds $P(\mu_n^*, \nu_n^*)$ are given in Table \ref{Tab Uni 0134 exyy}.
        \begin{table}[!htbp]\scriptsize
\centering\renewcommand{\arraystretch}{1.2}

      \begin{tabular}{|c|c|c|c|c|c|c|c|c|c|}
      \hline
      n & 3 & 4 & 5 & 6 & 7 & 8 & 9 & 10 & 11 \\
      \hline
      $P(\mu_n^*,\nu_n^*)$ &	65.8620&	62.7911&	62.0990&	61.9338&	61.8934&	61.8834&	61.8810&	61.8803&	61.8802
 \\ \hline
      $d_n$ &	254.8&	233.2&	224.1&	219.8&	217.7&	216.7&	216.1	&215.8	&215.5

 \\ \hline

       \end{tabular}
           \caption{Approximation results in the case 1.b)}
       \label{Tab Uni 0134 exyy}
     \end{table}
            \end{enumerate} 
            \item[2.] Let $\mu \sim \mathcal{U}[9,11], \nu \sim \mathcal{U}[0,20]$, i.e. $k=20$. Here, we partition the support of $\nu$ in maximally $2048$ intervals, i.e. we have a difference of $\frac{5}{512}$ between two partition points.
     \begin{enumerate}
       \item $c\left(x,y\right)=xy^2$. The real upper price bound is \begin{align*}
      &\sup_{\Q \in \mathcal{M}\left(\mu,\nu\right)}\Eop_{\Q}\left[c\left(X,Y\right) \right] =\Eop_{\mu}\left[X\left(\frac{11}{20}\left(\frac{11}{2}X-\frac{81}{2}\right)^2+\frac{9}{20}\left(\frac{99}{2}-\frac{9}{2}X\right)^2\right)\right] = 1356.5 .
       \end{align*}
       For all $n=3,\ldots,11$, we calculate the approximate upper price bound $P(\mu_n^*, \nu_n^*)$ and $d_n.$ This yields the results of Table \ref{Tab Uni 091120 xyy}.
        \begin{table}[!htbp]\scriptsize
\centering
\renewcommand{\arraystretch}{1.2}
      \begin{tabular}{|c|c|c|c|c|c|c|c|c|c|}
      \hline
      n  & 3 & 4 & 5 & 6 & 7 & 8 & 9 & 10 & 11 \\
      \hline
      $P(\mu_n^*,\nu_n^*)$	&1421	&1367.2&	1359.35&	1357.206	& 1356.676&	1356.543&	1356.511&	1356.503&	1356.501

 \\ \hline
      $d_n$	&4133&	2739	&2922&	2894&	2882	& 2805&	2826&	2808&	2819

 \\ \hline

       \end{tabular}
           \caption{Approximation results in the case 2.a)}
     \label{Tab Uni 091120 xyy}
     \end{table}




     \item $c\left(x,y\right)=\exp(x)\cdot y^2$. The real upper price bound is $ 4041627.609.$ The approximate upper price bounds $(\mu_n^*, \nu_n^*)$ are given in Table \ref{Tab Uni 091120 exyy}.
        \begin{table}[!htbp]\tiny
\centering
\renewcommand{\arraystretch}{1.2}
      \begin{tabular}{|c|c|c|c|c|c|c|c|c|}
      \hline
      n  
      & 4 & 5 & 6 & 7 &8 & 9 & 10 & 11  \\
      \hline
      $P(\mu_n^*,\nu_n^*)$ 	
      &	4826637&	4236165&	4093466	& 4054268 &4044652	&4042391&	4041818	&4041675	
 \\ \hline
      $d_n$ 	&	
      	200962386	& 199206022&	212331405&	207103361 &	198219006&	200101331&	199976026&	200114438
 \\ \hline
       \end{tabular}
           \caption{Approximation results in the case 2.b)}
       \label{Tab Uni 091120 exyy}
     \end{table}
            \end{enumerate}
\end{enumerate}

\newpage
\section*{Appendix}
\subsection{Proof of Theorem \ref{Thm Wasserstein Distance Estimates}}
 \begin{proof}
  By Remark \ref{rem:W}, we have
    \[
        W( \mu,\mu^{*} )=\int_{-\infty}^{\infty} \left| F_\mu(t)-F_{\mu^*}(t)\right| d t.
        \]
        In order to calculate the integral, we plug in the distribution function representations  \eqref{eq:BL} using the call option price function $C_{\mu}$ and use Lemma \ref{lem:mu*} to obtain $F_{\mu^*}$. Then we have
      \[
        W( \mu,\mu^{*} ) = \int_{0}^{K}\left|1+C_\mu'(t+)- F_{\mu^*}(t) \right| d t
        = \sum_{j=0}^{2^{n}-1}\int_{k_j}^{k_{j+1}} \left| C_{\mu}'(t+)-\frac{C_{\mu}(k_{j+1})-C_{\mu}(k_j)}{k_{j+1}-k_j} \right| d t.
         \]
         Note that $C_\mu^*(k_j)=C_\mu(k_j)$ for $j=0,\ldots 2^n$.
         In the following, let us abbreviate $m_{j}:= \frac{C_{\mu}(k_{j+1})-C_{\mu}(k_j)}{k_{j+1}-k_j}$.
    Since $F_{\mu^*}$ is constant on $[k_j,k_{j+1})$ for all $j=0,\ldots,2^{n}-1$ there is a $\theta_j\in [k_j,k_{j+1}]$ such that for all $t \in [k_j,\theta_j)$, we have $F_{\mu}(t) \leq F_{\mu^*}(t),$ or equivalently $C_{\mu}'(t+)\leq m_j$, and for all $t \in [\theta_j,k_{j+1})$, we have $ F_{\mu}(t) \geq F_{\mu^*}(t),$ or equivalently $C_{\mu}'(t+)\geq m_j$. Thus, we have
    \begin{equation}\label{Thm Proof Wasserstein Distance Estimates Origin}
      W( \mu,\mu^{*} ) = \sum_{j=0}^{2^{n}-1}\left[ \int_{k_j}^{\theta_j}\left( m_{j} - C_{\mu}'(t+)\right) d t + \int_{\theta_j}^{k_{j+1}}\left( C_{\mu}'(t+)-m_{j}\right) d t \right].
    \end{equation}
   We stress that the set of points $t \in \R_{+}$ such that $C_{\mu}'(t-)\neq C_{\mu}'(t+)$ is a Lebesgue null set. Hence, integrating over the right derivative $C_{\mu}'(\cdot+)$, we receive $C_{\mu}(\cdot)$. Based on \eqref{Thm Proof Wasserstein Distance Estimates Origin}, we thus obtain
    \begin{align*}
    W(\mu,\mu^{*}) &= \sum_{j=0}^{2^n-1} \bigg[ m_{j} \left(\theta_j-k_{j}\right) - \left( C_{\mu}(\theta_j)-C_{\mu}(k_{j}) \right)
   + \left( C_{\mu}(k_{j+1})-C_{\mu}(\theta_j) \right)-m_{j}\left( k_{j+1}-\theta_j \right)\bigg]
    \\
    &= \sum_{j=0}^{2^n-1} \Bigg[ \frac{C_{\mu}(k_{j+1})-C_{\mu}(k_{j})}{k_{j+1}-k_{j}}\left(\theta_j-k_{j} \right) - \left( C_{\mu}(\theta_j)-C_{\mu}(k_{j}) \right)
    \\
    & \quad \quad + \left( C_{\mu}(k_{j+1})-C_{\mu}(\theta_j) \right)-\frac{C_{\mu}(k_{j+1})-C_{\mu}(k_{j})}{k_{j+1}-k_{j}}\left( k_{j+1}-\theta_j \right)\Bigg]
    \\
    &= \sum_{j=0}^{2^n-1}  \frac{1}{k_{j+1}-k_{j}} \bigg[\left( C_{\mu}(k_{j+1})-C_{\mu}(k_{j}) \right)\left(\theta_j-k_{j} \right)
     - \left( C_{\mu}(\theta_j)-C_{\mu}(k_{j}) \right)\left( k_{j+1}-k_{j} \right)
    \\
    &  \quad + \left( C_{\mu}(k_{j+1})-C_{\mu}(\theta_j) \right)\left( k_{j+1}-k_{j} \right)
      -\left( C_{\mu}(k_{j+1})-C_{\mu}(k_{j}) \right)\left( k_{j+1}-\theta_j\right)\bigg].
    \end{align*}
         If we now add a suitable zero and rearrange the terms, then we obtain
    \begin{align*}
    W(  \mu,\mu^{*} )&= \sum_{j=0}^{2^n-1}\frac{2}{k_{j+1}-k_{j}}\bigg[\left( C_{\mu}(k_{j+1})-C_{\mu}(\theta_j) \right)\left(\theta_j-k_{j} \right)
     - \left( C_{\mu}(\theta_j)-C_{\mu}(k_{j}) \right)\left( k_{j+1}-\theta_j \right) \bigg]
    \\
    &=2\sum_{j=0}^{2^n-1}\lambda_jC_{\mu}(k_{j+1}) + (1-\lambda_{j})C_{\mu}(k_{j}) - C_{\mu}\left( \theta_j\right)
    \\
    &=2 \sum_{j=0}^{2^{n}-1}C_{\mu^{*}}(\theta_j)- C_{\mu}(\theta_j),
    \end{align*}
    where we use $\lambda_{j}:= \frac{\theta_j-k_{j}}{k_{j+1}-k_{j}}$ and the linearly interpolating definition of $C_{\mu^{*}}$. By the choice of $\theta_j$, we have that the slope of the line through $C_\mu(k_j)$ and $C_{\mu}(k_{j+1})$ is  in $[C_{\mu}'(\theta_j-),C_{\mu}'(\theta_j+)]$, i.e. it equals $C_\mu'(\theta_j)$ whenever the derivative exists. In particular, the distance of $C_{\mu^{*}}$ and $C_{\mu}$ on $[k_j,k_{j+1})$ is maximal in $\theta_j$. That is,
    \[
        \theta_j=argmax_{t \in [k_{j}, k_{j+1})}\left|C_{\mu^*}(t)-C_{\mu}(t)\right|.
    \]
    Thus, we have the desired representation
    \[
      W( \mu,\mu^{*} )=2\cdot \sum_{j=0}^{2^n-1}\sup_{t\in [k_j,k_{j+1})}|C_{\mu_{n}^{*}}(t)-C_{\mu}(t)|.
    \]

    Now we turn to the estimate in \eqref{Thm Eq 1 Wasserstein Distance Estimates}. For this purpose, we consider the slope  $C_{\mu}'(t+)$ for $t \in [k_j,k_{j+1})$. In particular, we have
    \[
          C_{\mu}'(t+)
      \begin{cases}
        \geq C_{\mu}'(k_j+),  & t \in [k_j,\theta_j)        \\
        \leq C_{\mu}'(k_{j+1}+), & t \in [\theta_j,k_{j+1}). 
      \end{cases}
        \]
        Using equation \eqref{Thm Proof Wasserstein Distance Estimates Origin}, we get
    \begin{align}\label{Thm Proof Wasserstein Distance Estimates Case Differentiable}
  \notag     W( \mu,\mu^{*} )
        &\leq \sum_{j=0}^{2^{n}-1}\left[ \int_{k_j}^{\theta_j}\left( m_{j}- C_{\mu}'(k_{j}+)\right) d t + \int_{\theta_j}^{k_{j+1}}\left(C_{\mu}'(k_{j+1}+)-m_{j}\right) d t \right]
      \\ \notag
      &= \sum_{j=0}^{2^{n}-1}\Bigg[\int_{k_j}^{\theta_j}\left( \frac{C_{\mu}(k_{j+1})-C_{\mu}(k_{j})}{k_{j+1}-k_{j}}- C_{\mu}'(k_{j}+)\right) d t
      \\ \notag       &\quad 
    + \int_{\theta_j}^{k_{j+1}}\left(C_{\mu}'(k_{j+1}+)-\frac{C_{\mu}(k_{j+1})-C_{\mu}(k_{j})}{k_{j+1}-k_j}\right) d t\Bigg]
      \\ \notag
      &= \sum_{j=0}^{2^{n}-1}\Bigg[\left( \frac{C_{\mu}(k_{j+1})-C_{\mu}(k_{j})}{k_{j+1}-k_{j}}- C_{\mu}'(k_{j}+)\right) \left(\theta_j-k_{j}  \right)
      \\
      &\quad + \left(C_{\mu}'(k_{j+1}+)-\frac{C_{\mu}(k_{j+1})-C_{\mu}(k_{j})}{k_{j+1}-k_{j}}\right)\left( k_{j+1}-\theta_j\right)\Bigg].
    \end{align}
 When we apply the inequalities $\theta_j \leq k_{j+1}$ and $-\theta_j \leq - k_j$ in \eqref{Thm Proof Wasserstein Distance Estimates Case Differentiable} (note that the terms in brackets are non-negative), then we obtain
    \begin{align}\label{Thm Proof Wasserstein Distance No Taylor}\notag
      W( \mu,\mu^{*} ) &\leq \sum_{j=0}^{2^{n}-1}\left( C_{\mu}'(k_{j+1}+)- C_{\mu}'(k_{j}+)\right) \left( k_{j+1}-k_{j}  \right)
      \\
      & = \frac{K}{2^n}\sum_{j=0}^{2^{n}-1}\left( C_{\mu}'(k_{j+1}+)- C_{\mu}'(k_{j}+)\right) \leq \frac{K}{2^n}.
    \end{align}

    In order to obtain the estimate in \eqref{Thm Eq 3 Wasserstein Distance Estimates}, we use the fact that the slopes get closer and closer when $n$ increases. We assume that $C_{\mu} \in C^2(\R_+)$ and rewrite the right hand side of \eqref{Thm Proof Wasserstein Distance Estimates Case Differentiable}. Then we have
    \begin{align*}
       W( \mu,\mu^{*} ) &\leq \sum_{j=0}^{2^{n}-1}\Bigg[\left( C_{\mu}(k_{j+1})-C_{\mu}(k_{j})- C_{\mu}'(k_{j})(k_{j+1}-k_{j})\right) \left( \frac{\theta_j-k_{j}}{k_{j+1}-k_{j}}  \right)
      \\
      & \quad  + \left(C_{\mu}'(k_{j+1})(k_{j+1}-k_{j})-C_{\mu}(k_{j+1})+C_{\mu}(k_{j})\right)\left( \frac{k_{j+1}-\theta_j}{k_{j+1}-k_{j}}\right)\Bigg].
    \end{align*}
    Now let us use the Theorem of Taylor. In particular, for two times continuously differentiable functions $f:\R \to \R$, we obtain
        \[
             f(x)= f(a)+  f'(a)(x-a)  + \int_{a}^{x}(x-t)f''(t) d t. 
         \]
    If we now apply this formula in the form
    \[
        f(x)-f(a)-f'(a)(x-a)= \int_{a}^{x}(x-t)f''(t) d t =: R_1 f(x,a),
    \]
     for $f \equiv C_\mu$ with $x=k_{j+1}$ and $a=k_j$, and with $x=k_j$ and $a=k_{j+1}$, then we obtain
         \[
        W( \mu,\mu^{*} )\leq 
        \sum_{j=0}^{2^{n}-1}\left( \left( \frac{\theta_j-k_{j}}{k_{j+1}-k_{j}} \right)R_1C_{\mu}(k_{j+1},k_{j})+ \left( \frac{k_{j+1}-\theta_j}{k_{j+1}-k_{j}} \right)R_1C_{\mu}(k_{j},k_{j+1}) \right).
         \]
            The well-known general Taylor residual estimate states that we have
    \[
        |R_1 f(x,a)|\leq \sup_{\xi \in (a-r,a+r)}\left|\frac{f''(\xi)}{2}(x-a)^2\right|
    \]
    for all $x \in (a-r,a+r)$. Choosing $r=\frac{K}{2^n}+\varepsilon$, $\varepsilon>0$, we achieve
    \begin{align*}
      \left|R_1C_{\mu}\left( k_{j+1},k_{j} \right)\right| &\leq \sup_{t \in \left( k_{j-1}-\varepsilon,k_{j+1}+\varepsilon \right)}\left|\frac{C_{\mu}''(t)}{2}\left( k_{j+1}-k_{j} \right)^2 \right|
      \\
      &= \sup_{t \in \left( k_{j-1}-\varepsilon,k_{j+1}+\varepsilon \right)}\left|\frac{C_{\mu}''(t)}{2}\left( \frac{K}{2^n}\right)^{2}\right| \leq T_{\mu}\cdot  K^2 \cdot 2^{-(2n+1)}
    \end{align*}
    and analogously $|R_1C_{\mu}(k_{j},k_{j+1})| \leq T_{\mu}\cdot  K^2 \cdot 2^{-(2n+1)}.$ Thus, we get
    \begin{align*}
      W( \mu,\mu^{*} ) &\leq \sum_{j=0}^{2^n-1}\left(\frac{\theta_j-k_{j}}{k_{j+1}-k_{j}} +  \frac{k_{j+1}-\theta_j}{k_{j+1}-k_{j}} \right) T_{\mu}\cdot K^2 \cdot2^{-(2n+1)}
      \\
      &= 2^n\cdot T_{\mu} \cdot K^2 \cdot 2^{-(2n+1)} = \frac{T_{\mu}\cdot K^{2}}{2^{n+1}},
    \end{align*}
    which is the desired estimate and thus ends the proof.
  \end{proof}

\subsection{Speed of Convergence in Theorem \ref{Thm Convergence Speed General} cannot be improved. }

    In this example, we show that the convergence speed proven in Theorem \ref{Thm Convergence Speed General} is maximal in the sense that it can not be improved in general. For this purpose, we consider the two discrete measures
    \[
          \mu=\frac{1}{4}\delta_{1}+\frac{1}{2}\delta_{\frac{7}{3}}+\frac{1}{4} \delta_{3}\quad  \text{ and } \quad \nu=\frac{1}{4}\delta_{0}+\frac{1}{2}\delta_{\frac{7}{3}}+\frac{1}{4} \delta_{4}.
        \]
        These measures have mass $1$, mean $\frac{7}{3}$ and the call option price functions    \begin{align*}
      C_{\mu}(k)&=\left( \frac{13}{6}-k \right)1_{\{0 \leq k \leq 1\}}+\left( \frac{23}{12}-\frac{3}{4}k \right)1_{\left\{ 1 \leq k \leq \frac{7}{3} \right\}}+\left( \frac{3}{4}-\frac{1}{4}k \right)1_{\left\{ \frac{7}{3} \leq k \leq 3 \right\}},
      \\
      C_{\nu}(\ell)&=\left( \frac{13}{6}-\frac{3}{4}\ell \right)1_{\left\{ 0 \leq \ell \leq \frac{7}{3} \right\}}+\left( 1-\frac{1}{4}\ell \right)1_{\left\{ \frac{7}{3} \leq \ell \leq 4 \right\}}.
    \end{align*}
    We easily see that $C_{\mu}\leq C_{\nu}$ and thus $\mu \leq_c \nu$.  As payoff function we choose $c(x,y)=xy^2$. The upper price bounds can be computed as 
    \[
     \frac{3}{28}c\left( 1,\frac{7}{3} \right) + \frac{11}{28} c\left( \frac{7}{3},\frac{7}{3} \right) + \frac{1}{16} c\left( \frac{7}{3},4 \right) + \frac{3}{16}c(3,4)= \frac{913}{54}.
    \]

    In order to prove the optimality of the convergence speed, we calculate the approximating measures $\mu_{n}^{*}$ and $\nu_{n}^{*}$ and the associated price bounds for general $n \geq 3$. The measures have the structure
    \[
          \mu_n^{*}=\frac{1}{4}(\delta_{1}+\delta_3) + \mu_n^{r} \text{ and } \nu_n^{*}=\frac{1}{4}(\delta_0+\delta_{4}) + \nu_n^{r},
        \]
        for all $n \in \N$, where $\mu_n^{r}$ and $\nu_n^{r}$ are also measures with two atoms close to $\frac{7}{3}$ each. This follows from the determination technique of the approximating measures based on the associated piecewise linearly interpolated call option price functions $C_{\mu_n^{*}} $ and $ C_{\nu_n^{*}}$. Indeed, these deviate from the functions $C_{\mu}$ and $C_{\nu} $ only on the interval $\left( k_{j(n)},k_{j(n)+1} \right)$, where $j(n)$ is such that $k_{j(n)}<\frac{7}{3}< k_{j(n)+1}$. Consequently, $k_{j(n)}$ and $k_{j(n)+1}$ are the atoms of the residual measures.

    We determine the general structure of these values and denote $j=j(n) $ and $k_{j}=k_{j(n)}$ for the rest of the example. We have
    \[
      k_{j}< \frac{7}{3}<k_{j+1} \iff 4 \cdot \frac{j}{2^n} < \frac{7}{3} < 4\cdot\frac{j+1}{2^n} \iff j<\frac{7}{3}\cdot \frac{2^n}{4}< j+1.
    \]
       As $j \in \N$, we clearly have
    \[
      j=\left\lfloor \frac{7}{3} \cdot\frac{2^n}{4} \right\rfloor =
        \begin{cases}
          \frac{7}{3}\cdot\frac{2^n}{4}-\frac{1}{3}, & n \text{ even},
          \\
          \frac{7}{3}\cdot\frac{2^n}{4}-\frac{2}{3}, & n \text{ odd},
        \end{cases}
    \]
        from which we immediately get
    \[
      k_{j}=4\cdot\frac{j}{2^n} = \begin{cases}
         \frac{7}{3} - \frac{4}{3\cdot 2^n}, & n \text{ even},\\
         \frac{7}{3} - \frac{8}{3\cdot 2^n}, & n \text{ odd},
       \end{cases} \quad  \text{ and }\quad  k_{j+1}=4\cdot\frac{j+1}{2^n} = \begin{cases}
         \frac{7}{3} + \frac{8}{3\cdot 2^n}, & n \text{ even}, \\
         \frac{7}{3} + \frac{4}{3\cdot 2^n}, & n \text{ odd}.
      \end{cases}
    \]

    The masses of $\mu_n^{*}$ and $\nu_n^{*}$ in the atoms $\delta_{k_j}$ and $\delta_{k_{j+1}}$ are the differences of the slopes of $C_{\mu_n^{*}}$ and $C_{\nu_n^{*}}$ on the intervals $(k_{j},k_{j+1})$ and $(k_{j-1},k_{j})$, and $(k_{j+1},k_{j+2})$ and $(k_{j},k_{j+1})$ respectively. As $C_{\mu}$ and $C_{\nu}$ have the same slopes in these areas, we know that the masses $  \omega_j^{n}, \vartheta_j^{n}  $ of the atoms are equal for $\mu_n^{*}$ and $\nu_n^{*}$. Hence, we have
    \[
          \omega_j^{n}= \vartheta_j^{n} = m_{j}^{n}-\left( -\frac{3}{4} \right) \text{ and } \omega_{j+1}^{n}= \vartheta_{j+1}^{n} = -\frac{1}{4}- m_{j}^{n},
        \]
        where $m_{j}^{n}=\frac{C_{\mu}(k_{j+1})-C_{\mu}(k_{j})}{k_{j+1}-k_{j}} = \frac{2^n}{4}(C_{\mu}(k_{j+1})-C_{\mu}(k_{j})) = \frac{2^n}{4}(C_{\nu}(k_{j+1})-C_{\nu}(k_{j})) .$ Using the representations of $C_{\mu}$ and $C_{\nu}$, we deduce
    \begin{align*}
      m_{j}^{n}&=\frac{2^n}{4}\left( 1-\frac{1}{4}k_{j+1}-\left( \frac{13}{6}-\frac{3}{4}k_{j} \right) \right)
      \\&= \frac{2^n}{4}\left( -\frac{7}{6}-\frac{1}{4}(k_{j+1}-k_{j})+\frac{1}{2}k_{j} \right)
      \\
      &= \frac{2^n}{4}\left( \frac{1}{2}k_{j}-\frac{7}{6} \right)-\frac{1}{4}
      \\
       &=
      \begin{cases}
        \frac{2^n}{4}\left( \frac{1}{2}\left( \frac{7}{3}-\frac{4}{3\cdot 2^n}-\frac{7}{6} \right) \right)-\frac{1}{4}=   -\frac{5}{12} , & n \text{ even},
        \\
        \frac{2^n}{4}\left( \frac{1}{2}\left( \frac{7}{3}-\frac{8}{3\cdot 2^n}-\frac{7}{6} \right) \right)-\frac{1}{4}=  -\frac{7}{12} , & n \text{ odd.}
       \end{cases}
    \end{align*}
   This finally implies
    \[
      \omega_j^{n}=\vartheta_{j}^{n}=\begin{cases}
         \frac{1}{3}, & n \text{ even}, \\
         \frac{1}{6}, & n \text{ odd},
      \end{cases} \quad  \text{ and }  \quad \omega_{j+1}^{n} = \vartheta_{j+1}^{n}=\begin{cases}
        \frac{1}{6}, & n \text{ even}, \\
         \frac{1}{3}, & n \text{ odd}.
      \end{cases}
    \]

    In total, we have the general structure
      \begin{align*}
      \mu_n^{*}&=\frac{1}{4}(\delta_{1}+\delta_3) + \frac{1}{6}\left( \delta_{k_{j}}+\delta_{k_{j+1}} \right) + \frac{1}{6}\left( \delta_{k_j}1_{\left\{ n \text{ even} \right\}} + \delta_{k_{j+1}}1_{\left\{ n \text{ odd} \right\}} \right),
      \\
     \nu_n^{*}&=\frac{1}{4}(\delta_{0}+\delta_4) + \frac{1}{6}\left( \delta_{k_{j}}+\delta_{k_{j+1}} \right) + \frac{1}{6}\left( \delta_{k_j}1_{\left\{ n \text{ even} \right\}} + \delta_{k_{j+1}}1_{\left\{ n \text{ odd} \right\}} \right).
      \end{align*}

    By construction we have $\mu_{n}^{*} \leq_c \nu_{n}^{*}$. 
The approximate upper price bound is then given by
     \begin{align*}
      & c\left(1,k_j\right)\cdot \frac{1}{4k_j}+ c\left(k_j,k_j\right)\cdot\left(\frac{1}{3}-\frac{1}{4k_j}\right)+c\left(k_j,k_{j+1}\right)\cdot\frac{1}{4 k_{j+1}}
        \\&+ c\left(k_{j+1},k_{j+1}\right)\cdot\left(\frac{1}{6} - \frac{1}{4 k_{j+1}}\right) +c\left(k_{j+1},4\right)\cdot\frac{1}{16}+ c\left(3,4\right)\cdot\frac{3}{16}
        \\ =& \frac{k_j}{4}+\frac{k_j^{3}}{3}-\frac{k_j^{2}}{4} + \frac{k_j\cdot k_{j+1}}{4}+ \frac{k_{j+1}^3}{6}-\frac{k_{j+1}^{2}}{4} + k_{j+1}+9.
       \end{align*}
       Plugging in the derived representations of $k_j$ and $k_{j+1}$, we get
       \begin{align*}
     &9 + \frac{7}{12}-\frac{1}{3\cdot 2^n} + \frac{7}{3}+\frac{8}{3\cdot 2^n}+ \frac{k_{j}^{3}}{3} + \frac{k_{j+1}^{3}}{6} + \frac{k_j \cdot k_{j+1}- k_{j}^2 - k_{j+1}^2}{4}.
        \end{align*}
        If we now also plug in the representations for the higher degree terms and rearrange the former, then we achieve
        \begin{align*}
  & \frac{143}{12} + \frac{7}{3 \cdot 2^n}+ \frac{1}{3}\left(\left(\frac{7}{3}\right)^{3}-\left( \frac{7}{3}\right)^{2}\frac{4}{3 \cdot 2^n} + \frac{7}{3} \cdot \frac{16}{9 \cdot 2^{2n}}- \frac{64}{27 \cdot 2^{3n}}\right)
        \\&+ \frac{1}{6}\left(\left(\frac{7}{3}\right)^{3}+\left( \frac{7}{3}\right)^{2}\frac{8}{3 \cdot 2^n} + \frac{7}{3} \cdot \frac{64}{9 \cdot 2^{2n}}+ \frac{512}{27 \cdot 2^{3n}}\right)
        \\ &+ \frac{1}{4}\bigg(\frac{49}{9}+\frac{28}{9 \cdot 2^{n}}-\frac{32}{9 \cdot 2^{2n}} -\left(\frac{49}{9} - \frac{56}{9 \cdot 2^{n}} + \frac{16}{9 \cdot 2^{2n}}\right)
        \\
        & \quad - \left(\frac{49}{9}+\frac{112}{9 \cdot 2^n} + \frac{64}{9 \cdot 2^{2n}}\right)\bigg)
        \\ =& \frac{913}{54} + \frac{84}{54 \cdot 2^n} + \mathcal{O}\left(\frac{1}{2^{2n}}\right).
     \end{align*}
     Analogously, for $n\geq 3$ odd, we have
     $$ \frac{913}{54} + \frac{78}{54 \cdot 2^n} + \mathcal{O}\left( \frac{1}{2^{2n}} \right).$$
     This is indeed the convergence speed from Theorem \ref{Thm Convergence Speed General}.

\end{document}